\documentclass[UKenglish,cleveref, autoref, thm-restate,a4paper]{article}
\pdfoutput=1
\usepackage{arxiv}

\usepackage{url}
\usepackage{hyperref}
\hypersetup{
    colorlinks=true,
    linkcolor=blue,
    filecolor=magenta,      
    urlcolor=red,
    citecolor=blue,
    pdftitle={Overleaf Example},
    pdfpagemode=FullScreen,
    }
    
\usepackage[margin=1in]{geometry} 

\usepackage{amssymb,amsmath,amsthm, amsfonts}
\usepackage{qcircuit}
\usepackage{thmtools}
\usepackage{thm-restate}
\usepackage{tikz}
\usepackage{cleveref}
\usepackage{float}
\usepackage{graphicx}
\usepackage{float}
\usepackage{color}
\usepackage{braket}
\usepackage{xcolor}
\usepackage{tablefootnote}
\usepackage{graphicx}
\usepackage{amsmath}
\usepackage{amssymb}
\usepackage{appendix}
\usepackage{amsthm}
\usepackage{amsfonts}
\usepackage{tcolorbox}
 \tcbuselibrary{skins , breakable , listings ,theorems}
\usepackage{colortbl}
\usepackage{comment}
\newtheorem{Fact}{Fact}

\newtheorem{theorem}{Theorem}[section]
 
\newtheorem{lemma}[theorem]{Lemma}

\theoremstyle{definition}
\newtheorem{definition}[theorem]{Definition}
\newtheorem{remark}[theorem]{Remark}

\newtheorem{conjecture}[theorem]{Conjecture}
\newtheorem{claim}[theorem]{Claim}

\usepackage[linesnumbered,ruled,vlined]{algorithm2e}

\usepackage{bm}

\newcommand{\inbrace}[1]{\left \{ #1 \right \}}
\newcommand{\inparen}[1]{\left ( #1 \right )}

\newcommand{\fnk}{f_n^k}
\newcommand{\fnkquanbias}{\frac{k}{n}\cdot T^2}
\newcommand{\fnkclassbias}{\frac{k}{n}\cdot T}

\newcommand{\kl}{f_n^{k,l}}

\newcommand{\poly}{\mathrm{poly}}

\newcommand{\IAdvf}{\frac{l-k}{\sqrt{(n-k)l}}}
\newcommand{\IAdvfsqua}{\frac{(l-k)^2 }{(n-k)l}}

\newcommand{\diff}{\frac{l-k}{n}}

\renewcommand{\hat}{\widehat}

\newcommand{\B}{\{0,1\}}

\let\OldLambda\lambda
\let\lambda\relax

\DeclareMathOperator{\lambda}{\OldLambda}
\DeclareMathOperator{\FC}{\mathsf{FC}}

\DeclareMathOperator{\Q}{\mathsf{Q}}

\DeclareMathOperator{\bs}{\mathsf{bs}}
\DeclareMathOperator{\fbs}{\mathsf{fbs}}
\DeclareMathOperator{\adeg}{\mathsf{\widetilde{deg}}}
\let\deg\relax
\DeclareMathOperator{\deg}{\mathsf{deg}}

\title{On the Fine-Grained Query Complexity of Symmetric Functions}

\date{}

\author{Supartha Podder\thanks{supported by US National Science Foundation (award no 1954311).}\\ Department of Computer Science, 
Stony Brook University, New York, USA\\ \textit{supartha@cs.stonybrook.edu} 
\And
	Penghui Yao\thanks{supported by National Natural Science Foundation of China (Grant No. 62332009, 61972191) and Innovation Program for Quantum Science and Technology (Grant No. 2021ZD0302900).}\\ State Key Laboratory for Novel Software Technology, Nanjing University, Nanjing, China \\ Hefei National Laboratory, Hefei, China \\ \textit{phyao1985@gmail.com} 
\And
Zekun Ye\thanks{supported by National Natural Science Foundation of China (Grant No. 62332009, 61972191) and Innovation Program for Quantum Science and Technology (Grant No. 2021ZD0302900).}\\ State Key Laboratory for Novel Software Technology, Nanjing University, Nanjing, China\\ \textit{yezekun@smail.nju.edu.cn} 
}

\begin{document}

\maketitle

\begin{abstract}
Watrous conjectured that the randomized and quantum query complexities of symmetric functions are polynomially equivalent, which was resolved by Ambainis and Aaronson \cite{AA14}, and was later improved in \cite{Cha19, BCG+20}. This paper explores a fine-grained version of the Watrous conjecture, including the randomized and quantum algorithms with success probabilities arbitrarily close to $1/2$. Our contributions include the following:
\begin{enumerate}

    \item An analysis of the optimal success probability of quantum and randomized query algorithms of two fundamental partial symmetric Boolean functions given a fixed number of queries. We prove that for any quantum algorithm computing these two functions using $T$ queries, there exist randomized algorithms using $\mathsf{poly}(T)$ queries that achieve the same success probability as the quantum algorithm, even if the success probability is arbitrarily close to 1/2. These two classes of functions are instrumental in analyzing general symmetric functions.

    \item We establish that for any total symmetric Boolean function $f$, if a quantum algorithm uses $T$ queries to compute $f$ with success probability $1/2+\beta$, then there exists a randomized algorithm using $O(T^2)$ queries to compute $f$ with success probability $1/2+\Omega\inparen{\delta\beta^2}$ on a $1-\delta$ fraction of inputs, where $\beta,\delta$ can be arbitrarily small positive values. As a corollary, we prove a randomized version of Aaronson-Ambainis Conjecture \cite{AA14} for total symmetric Boolean functions in the regime where the success probability of algorithms can be arbitrarily close to 1/2. 

    \item  We present tight polynomial equivalences for several fundamental complexity measures of partial symmetric Boolean functions. Specifically, we first prove that for certain partial symmetric Boolean functions, quantum query complexity is at most quadratic in approximate degree for any error arbitrarily close to 1/2. Next, we show exact quantum query complexity is at most quadratic in degree. Additionally, we give the tight bounds of several complexity measures, indicating their polynomial equivalence. Conversely, we exhibit an exponential separation between randomized and exact quantum query complexity for certain partial symmetric Boolean functions.

\end{enumerate}
\end{abstract}

\newpage

\tableofcontents

\newpage

\section{Introduction}\label{sec:intro}

Exploring quantum advantages is a key problem in quantum computing. A lot of research work has revolved around analyzing and characterizing quantum advantages, such as \cite{BGKT19,GS20,CCHL21,Kallaugher21,YZ22}. 
Query complexity is a complexity model commonly used to describe quantum advantages, since it enables the establishment of both nontrivial upper and lower bounds. 
A comprehensive survey on the query complexity can be found in \cite{buhrman2002complexity}. A series of works \cite{Aaronson2018problem, Tal2020Towards,Bansal2021Forrelation,Sherstov2021An} has shown that for partial functions, quantum query complexity could be exponentially smaller (or even less) than the randomized query complexity, while for total functions, they are always polynomially related \cite{ABK+21}. Although the query complexity model has demonstrated the powerful ability of a quantum computer to solve certain ``structured'' problems more efficiently than a classical computer, such as Simon's problem \cite{simon1994power} and integer factorization problem \cite{shor1994algorithms}, there only exist at most quadratic quantum speedups for some ``unstructured'' problems, such as black-box search problems \cite{Grover96}. Thus, one natural question is to explore how much structure is needed for significant quantum speedups \cite{AA14}. 

Watrous conjectured that the randomized query complexity and quantum query complexity of partial symmetric functions are polynomial equivalent~\cite{AA14}. Later, Aaronson and Ambainis \cite{AA14} initiated the study of quantum speedup on the quantum query complexity of symmetric functions, showing that partial functions invariant under full symmetry do not exhibit super-polynomial quantum speedups, which resolves the Watrous conjecture. Their result was later improved by Chailloux~\cite{Cha19}, who achieved a tighter bound and removed a technical dependence of output symmetry.
Recently \cite{BCG+20} performed a systematic analysis of functions symmetric under other group actions and characterized when super-polynomial quantum speedups are achievable. 
However, all these results work in the bounded error regime and do not explicitly consider arbitrary small biases.

 In this paper, we propose and investigate a fine-grained version of the Watrous conjecture concerning the quantum and randomized query complexities of symmetric functions with an arbitrary error close to 1/2. Before stating the conjecture, we need to introduce two notions which are essential to this paper. For any Boolean function $f$ and $T>0$, let the classical $T$-bias $\delta_C(f,T)$ be the optimal success probability of $T$-query randomized algorithms minus $1/2$ over all possible inputs. The quantum $T$-bias $\delta_Q(f,T)$ is defined for quantum algorithms analogously.

\begin{conjecture}[Fine-grained Watrous conjecture]
    There exists a constant $c\geq 2$ satisfying that for any partial symmetric function $f$ and any $T>0$, $\delta_C(f, c\cdot T)\geq\frac{\delta_Q(f, T)}{\mathsf{poly}(T)}$.
\end{conjecture}

The reason for $c\geq 2$ is that the $n$-bit parity function can be exactly computed with $n/2$ quantum queries, while any randomized algorithm with less $n$ queries succeeds with probability $1/2$. It is also not hard to see the fine-grained Watrous conjecture implies that quantum and randomized query complexities of symmetric functions are polynomially related. Indeed, if $\delta_Q(f, T)$ is lower bounded by some constant, then $\delta_C(f, c\cdot T)\geq\Omega(\frac{1}{\mathsf{poly}(T)})$, which implies the randomized query complexity of $f$ is $\mathsf{poly}(T)$ by error reduction.

\subsection{Our Motivation and Contribution}
To study the fine-grained Watrous conjecture, we start with the following two fundamental symmetric Boolean functions, which are also essential to analyze general symmetric functions:
\begin{equation*}
\fnk(x) = 
\begin{cases}
0, & \text{ if }|x| \in \{0,n\},\\
1, & \text{ if }|x| \in \{k,n-k\}, \\
\end{cases}
\end{equation*}
where $1 \le k \le n/2$ and
\begin{equation*}
\kl(x) = 
\begin{cases}
0, & \text{ if }|x| = k, \\
1, & \text{ if }|x| = l,
\end{cases}
%(k < l)
\end{equation*}
where $k < l$.
The famous Deutsch-Jozsa problem \cite{DJ92} and the decision version of the unstructured search problem \cite{Grover96} can be interpreted as special cases of these two functions. Fefferman and Kimmel considered the subset-sized checking problem \cite{fefferman_et_al:LIPIcs:2018:9604, goldwasser1986private}, where the $n$-bit input string is promised to have either $\sqrt{n}$ or $0.99\sqrt{n}$ many marked items, and the goal is to decide which case. Their result together with \cite{goldwasser1986private} proved that it can be served as an oracle separation between $\mathsf{AM}$ and $\mathsf{QCMA}$. 
On a similar flavour, \cite{aaronson2019quantum} introduced the approximate counting problem, where the $n$-bit input has either $\leq w$ or $\geq 2w$ many marked items.  
Our function $\kl(\cdot)$ can be seen as a symmetric variant of these two problems\footnote{Our problem can also be seen as a Gap-Threshold function. Threshold function is defined as $f^k_n(x)=1$ iff $|x|\geq k$.}. We study the tradeoff between the number of queries and the optimal success probabilities in quantum and randomized settings for both functions with errors close to 1/2. 

We further consider the relation between various complexity measures of partial symmetric Boolean functions. For total Boolean functions, it has been proved that several fundamental complexity measures are polynomial equivalent (See Table 1 in \cite{ABK+21}). Moreover, the tight bounds on several fundamental complexity measures of total symmetric Boolean functions have also been obtained \cite{buhrman2002complexity}. However, the result for partial symmetric Boolean functions has not been fully characterized.

Our contribution is as follows. For convenience, if $f$ is a \textit{symmetric} Boolean function, we denote $f(k) = f(x)$ for any $|x| = k$. Additionally, we say an $n$-bit Boolean function $f:D\rightarrow \B$ is \textit{even} if $f(|x|) = f(n-|x|)$ for any $x\in D$, where $D \subseteq \B^n$.
\begin{enumerate}
    \item For both randomized setting and quantum setting, we characterize the optimal success probability of algorithms given the number of queries for the function $\fnk$ (\Cref{th:fnk1}) and $\kl$ (\Cref{th:kl_quantum,th:kl_class}). As a corollary, we show for any $T$-query quantum algorithm to compute $\fnk$ and $\kl$, there exist classical randomized algorithms using $\mathsf{poly}(T)$ queries to simulate the success probability of the quantum algorithm (\Cref{cor:fnk,cor:kl}). Additionally, we characterize the exact quantum query complexity of $\fnk$ (\Cref{th:fnk2}).

    \item We establish a relation between the number of queries and the bias of quantum and randomized algorithms to compute total symmetric Boolean functions, where the bias of the algorithms can be arbitrarily small (\Cref{th:relation_Q_R}). As a corollary, we prove a weak version of Conjecture \ref{conjecture:AA}:
the acceptance probability of a quantum
query algorithm to compute a total symmetric Boolean function can be approximated by a randomized algorithm with only a polynomial increase in the number of queries, where the bias of quantum algorithms can be arbitrarily small (Corollary \ref{AA_conjecture}).

    \item We investigate the relation between different complexity measures of partial symmetric Boolean functions. 
Specifically, Theorem \ref{th:poly_Q_adeg} shows the relation between the quantum query complexity and the approximate degree of even  partial symmetric Boolean functions for arbitrarily small bias\footnote{\Cref{th:poly_Q_adeg} is also a new result for total symmetric Boolean functions.}.
Theorem \ref{th:polyrelation1} shows exact quantum query complexity and degree are quadratically related.  
\Cref{th:polyrelation2} presents tight bounds of block sensitivity, fractional block sensitivity, quantum query complexity, and approximate degree, where quantum query complexity and approximate degree are in a bounded-error setting. \Cref{th:poly_Q_bs} shows block sensitivity is an upper bound of quantum query complexity. Since it has been known that $Q(f) \ge \Omega(\sqrt{\bs(f)})$ for any (possibly partial) Boolean function $f$ \cite{Ben2021polynomial}, quantum query complexity and block sensitivity of partial symmetric Boolean functions are polynomially related.
On the converse, Theorem \ref{th:exprelation} shows
an exponential gap between the exact quantum query complexity and randomized query complexity for some partial symmetric Boolean functions, which is different from total symmetric functions.
\end{enumerate}

\begin{conjecture}[Aaronson-Ambainis Conjecture \cite{AA14}]\label{conjecture:AA}
The acceptance probability of a $T$-query quantum algorithm to compute a Boolean function can be approximated by a deterministic algorithm using $\poly(T,1/\epsilon,1/\delta)$ queries within an additive error $\epsilon$ on a $1-\delta$ fraction of inputs. 
\end{conjecture}

\begin{restatable}{theorem}{optfnk}
\label{th:fnk1}
For $T>0$, the quantum $T$-bias and  classical $T$-bias of $\fnk$ are   
\begin{equation*}
\delta_Q\inparen{\fnk, T}= 
\begin{cases}
\Theta(\fnkquanbias),&  \text{ if } T \le  \sqrt{n/k},\\
\Theta\inparen{1},&\text{ if } T > \sqrt{n/k},
\end{cases}
\end{equation*}
\begin{equation*}
\delta_C\inparen{\fnk, T}= 
\begin{cases}
0, & \text{ if }T=1, \\
\Theta(\fnkclassbias), & \text{ if }2 \le T \le n/k, \\
\Theta(1), & \text{ if }T > n/k.
\end{cases}
\end{equation*}
\end{restatable}
\begin{restatable}{theorem}{quanOptKl}
\label{th:kl_quantum}
For $\kl$ and $T>0$, the quantum  $T$-bias is
\begin{equation*}
\delta_Q\inparen{\kl, T}= 
\begin{cases}
\Theta\inparen{\min \inbrace{\frac{l-k}{\sqrt{(n-k)l}}\cdot T, \frac{l-k}{n} \cdot T^2}}, & \text{ if }T= O\inparen{\frac{\sqrt{(n-k)l}}{l-k}}, \\
\Theta(1), & \text{ if }T = \Omega\inparen{\frac{\sqrt{(n-k)l}}{l-k}}.
\end{cases}
\end{equation*}
\end{restatable}
\begin{restatable}{theorem}{classOptKl}
\label{th:kl_class}
If $T = O\inparen{\frac{(n-k)l}{(l-k)^2}}$, the classical $T$-bias of $\kl$ satisfies that
\[
\begin{aligned}
    \delta_C\inparen{\kl, T} &=O\inparen{\min \inbrace{\IAdvf \cdot\sqrt{T}+\frac{T}{n}, \diff \cdot T}},\\
\delta_C\inparen{\kl, T} &=
    \Omega\inparen{\max\inbrace{\IAdvfsqua \cdot T, \diff \cdot\sqrt{T}}}.
\end{aligned}
\]
If $T = \Omega\inparen{\frac{(n-k)l}{(l-k)^2}}$, then $\delta_C\inparen{\kl, T}= \Theta(1)$.
\end{restatable}

\begin{restatable}{corollary}{corfnk}
\label{cor:fnk}
For arbitrarily small bias $\beta > 0$, if there exists a quantum algorithm using $T$ queries to compute $\fnk$ with success probability $1/2+\beta$, then there also exists a classical randomized algorithm using $O(T^2)$ queries to compute $\fnk$
with the same success probability.
\end{restatable}

\begin{restatable}{corollary}{corKl}
\label{cor:kl}
For arbitrarily small bias $\beta > 0$, if there exists a quantum algorithm using $T$ queries to compute $\kl$ with success probability $1/2+\beta$, then there also exist classical randomized algorithms using $T^2$ queries to compute $\kl$ with success probability $1/2+\Omega\inparen{\beta^2}$ and using $T^4$ queries to compute $\kl$ with success probability $1/2+\Omega\inparen{\beta}$. Thus \[\delta_C\inparen{\kl,T^2}\geq\Omega\inparen{\delta_Q(\kl,T)^2}~\mbox{and}~\delta_C\inparen{\kl,T^4}\geq\Omega\inparen{\delta_Q(\kl,T)}.\]
\end{restatable}

\begin{restatable}{theorem}{exactfnk}
\label{th:fnk2}
The exact quantum query complexity of $\fnk$ satisfies $\lceil \frac{\pi}{2\theta} \rceil \le Q_E(f_n^k) \le \lceil \frac{\pi}{2\theta} \rceil+2$, where $\theta = 2\arcsin{\sqrt{k/n}}$, whereas the zero-error randomized query complexity of $\fnk$ is $n-k+1$.
\end{restatable}

\begin{restatable}{theorem}{relationQR}
\label{th:relation_Q_R}
For any total symmetric Boolean function $f$ and arbitrarily small bias $\beta > 0$, if there exists a quantum algorithm using $T$ queries to compute $f$ with success probability $1/2+\beta$, then for any $\delta \in (0,1)$, there exists a randomized algorithm using $O(T^2)$ queries to compute $f$ with success probability $1/2+\Omega\inparen{\delta\beta^2}$  on a $1-\delta$ fraction of inputs. 
\end{restatable}

\begin{restatable}{corollary}{AA}
    \label{AA_conjecture}
    For any total symmetric Boolean function $f$ and arbitrarily small bias $\beta > 0$, if there exists a $T$-query quantum algorithm  to compute $f$ with success probability $1/2+\beta$, then for any $\epsilon \in (0, \beta), \delta \in (0,1)$, there exists a randomized algorithm using 
    $O(T^2/(\epsilon^2\delta^2))$
    queries to compute $f$ with success probability $1/2+(\beta-\epsilon)$ on a $1-\delta$ fraction of inputs.
\end{restatable}

\begin{remark}
Corollary \ref{AA_conjecture} is a randomized version of Conjecture \ref{conjecture:AA}. Additionally, Corollary \ref{AA_conjecture} considers total symmetric Boolean functions, while Conjecture \ref{conjecture:AA} refers to any Boolean function.
\end{remark}
\begin{restatable}{theorem}{adegQ}
\label{th:poly_Q_adeg}
For any (possibly partial) symmetric Boolean function $f$ satisfying $f(x) = f(n-x)$ and arbitrarily small $\beta > 0$, if $T = \adeg_{\frac{1}{2}-\beta}(f)$, 
there exists a quantum query algorithm using $\lceil T/2 \rceil$ queries to compute $f$ with success probability $1/2+\Omega\inparen{\beta/\sqrt{T}}$. 
Namely, \[\delta_Q\inparen{f,\adeg_{\frac{1}{2}-\beta}(f)}=\Omega\inparen{\frac{\beta}{\sqrt{\adeg_{\frac{1}{2}-\beta}(f)}}}.\]
As a corollary, we have $Q_{\epsilon}(f) = O(\adeg_{\epsilon}(f)^2)$ for any error $\epsilon$ arbitrarily close to $1/2$. 
\end{restatable}
\begin{restatable}{theorem}{degQE}
\label{th:polyrelation1}
For any partial symmetric Boolean function $f$, we have $Q_E(f) = O(\deg(f)^2)$.
\end{restatable}
\begin{restatable}{theorem}{bsfbs}
\label{th:polyrelation2}
For any partial symmetric Boolean function $f$, we have 
\begin{equation*}
\begin{aligned}
\bs(f) &= \Theta\left(\fbs(f)\right)& & = \inparen{\max\limits_{k<l:f(k)\neq f(l)} \frac{n}{l-k}}, \\
Q(f) &= \Theta\left(\adeg(f)\right)& & =  \inparen{\max\limits_{k<l:f(k)\neq f(l)} \frac{\sqrt{(n-k)l}}{l-k}}.\\
\end{aligned}
\end{equation*}
\end{restatable}
\begin{restatable}{corollary}{Qbs}
\label{th:poly_Q_bs}
For any partial symmetric Boolean function $f$, we have $Q(f) = O\inparen{\bs(f)}$. 
\end{restatable}
\begin{theorem}\label{th:exprelation}
There exists a partial symmetric Boolean function $f$ such that $Q_E(f) = \Omega(n)$ and $R(f) = O(1)$.
\end{theorem}

\subsection{Proof Techniques}
In this section, we give a high-level technical overview of our main results. 

\subsubsection{Upper and Lower Bounds on Quantum $T$-Bias}
We use several methods to show the upper bound on the quantum $T$-bias of different symmetric Boolean functions:
\begin{enumerate}
    \item For $\fnk$, we show if the number of a quantum algorithm is no more than $T$ queries, then the bias $\beta$ of the algorithm is at most 
    $O(T^2/\bs(f_n^k))$ (See \Cref{eq:bias_bs}), where $\bs(f_n^k)$ is the block sensitivity of $f_n^k$. By solving a lower bound of $\bs(f_n^k)$, we obtain an upper bound on the quantum $T$-bias of $f_n^k$ (\Cref{th:fnk1}). 
    \item For $\kl$, using Paturi's lower bound technique \cite{Paturi92} for the approximate degree of symmetric Boolean functions, we give the following lower bound (See Fact \ref{beals} and \Cref{lemma:deg_lowerbound}):
    \[
    Q_{\epsilon}(f_n^{k,l}) \ge \frac{1}{2}
    \adeg_{\epsilon}\left(\kl\right) = \Omega\inparen{\max\inbrace{\frac{\beta\sqrt{\left(n-k\right)l}}{l-k}
, \sqrt{\frac{\beta n}{l-k}}}},
\]
where $\beta = 1/2-\epsilon$. 
The quantum $T$-bias of $f_n^{k,l}$ is derived by this lower bound
    (\Cref{th:kl_quantum}).
\end{enumerate}

To obtain the lower bound on the quantum $T$-bias, we also use diverse ideas to design $T$-query quantum algorithms:
\begin{enumerate}
    \item For $\fnk$ and $\kl$, we use various variants of amplitude amplification algorithm and analyze the success probability of algorithms meticulously (\Cref{th:fnk1,th:kl_quantum}).
\item For even symmetric Boolean functions, we design a novel quantum algorithm by taking advantage of the Chebyshev expansion and constructing controlled Grover's diffusion operations (\Cref{th:poly_Q_adeg}).
\end{enumerate}

\subsubsection{Upper and Lower Bounds on Classical $T$-Bias}
For $\fnk$ and $\kl$, we show the upper bound on the classical $T$-bias by analyzing the total variation distance of distributions; for the lower bound, we give sampling algorithms to estimate Hamming weights of the input and analyze the success probability of the algorithms (\Cref{th:fnk1,th:kl_class}). 

For the lower bound on the classical $T$-bias of total symmetric Boolean functions, we 
design an innovative randomized algorithm by utilizing the Kravchuk polynomial when the number of queries is $T$. The analysis of the algorithm also makes use of the orthogonality property of the Kravchuk polynomial (\Cref{th:relation_Q_R}).

\subsubsection{The Relation Between Complexity Measures}
The key ideas to build the relation between complexity measures of partial symmetric Boolean functions are as follows:
\begin{enumerate}
    \item 
    In \Cref{th:polyrelation1},
    we show the relation between the exact quantum query complexity and the degree by giving the lower bound of the degree and designing a matching exact quantum algorithm up to a polynomial level. 
    Similar to the proof of \Cref{th:fnk2}, the exact quantum algorithm makes use of a subroutine to distinguish $|x| = k$ from $|x| = l$ exactly \cite{HSY+18}.
    \item 
    In \Cref{th:polyrelation2}, the analysis of block sensitivity and fractional block sensitivity relies on the symmetry property of the function. Furthermore, we show the quantum query complexity and the approximate degree of any partial symmetric Boolean function $f$ are equivalent to a constant factor. While the lower bound is well known (\Cref{beals}), we show $Q(f) \le \adeg(f)$ by giving a quantum approximate counting algorithm using $O(\adeg(f))$ quantum queries.
    \item 
    The exponential gap in \Cref{th:exprelation} is shown by giving a function easy to compute in a bounded-error case but has a large degree. 
\end{enumerate}

\subsection{Related Work}
The need for structure in quantum speedups has been studied extensively. Beals, Buhrman, Cleve, Mosca and de Wolf \cite{BBC+01} showed that there exists at most polynomial quantum speedups for total Boolean functions in the query model. Thus, the exponential speedups may only occur at partial functions. Furthermore, Aaronson and Ambainis \cite{AA14} showed that symmetric functions do not allow super-polynomial quantum speedups, even if the functions are partial. Chailloux \cite{Cha19} improved this result for a broader class of symmetric functions. Ben-David, Childs, Gilyén, Kretschmer, Podder and Wang \cite{BCG+20} further showed that hypergraph symmetries in the adjacency matrix model allow at most polynomial separations between quantum and randomized query complexities. Ben-David \cite{BenDavid16} proved a classical and quantum polynomial equivalence for a class of functions satisfying a certain symmetric promise. Aaronson and Ben-David \cite{AB16} showed that there exists at most polynomial quantum speedups to compute an $n$-bit partial Boolean function if the domain 
$D = \poly(n)$. Nonetheless, all these results concern the algorithms with a constant probability of success. They do not cover the query complexity with a subconstant probability of success. 

We also survey some results about the optimal success probability of quantum algorithms when the number of queries is fixed. 
For the unstructured search problem, Zalka \cite{Zalka} showed an optimal success probability of a quantum algorithm given the number of queries. For the collision finding problem, Zhandry \cite{Zhandry15} gave the upper bound on the success probability of quantum algorithms when the number of queries is fixed,  which matched the algorithm proposed by Brassard, H{\o}yer and Tapp \cite{BHT98}. 
Ambainis and Iraids \cite{AI16} analyzed the optimal success probability of one-query quantum algorithms to compute $\text{EQUALITY}_n$ and $\text{AND}_n$ functions. Montanaro, Jozsa, and Mitchison \cite{MJM15} indicated the optimal success probability of small symmetric Boolean functions when given any number of queries by numerical results. There is not much study about the optimal success probability with a given number of queries for symmetric Boolean functions.
Our work will fill the gap in this field. 

For a nonconstant $n$-bit total symmetric Boolean function $f$, it has been known that some fundamental complexity measures are $\Theta(n)$,
and $Q(f) = \Theta\inparen{\adeg(f)} = \Theta\inparen{\sqrt{n(n-\Gamma(f))}}$, 
where $\Gamma(f) = \min\inbrace{|2k-n+1|:f(k)\neq f(k+1)}$ \cite{buhrman2002complexity}. Mittal, Nair and Patro \cite{MNP21} further explored the best possible separation between sensitivity, block sensitivity and certificate complexity of total symmetric functions. Sherstov \cite{Sherstov09appro} gave an almost tight characterization of $\deg_{\epsilon}(f)$ for specific $\epsilon \in [1/2^n,1/3]$. Afterward, de Wolf \cite{Wolf10} obtained the optimal bound. Regarding the complexity measures of partial symmetric Boolean functions, Aaronson and Ambainis \cite{AA14} showed for any partial symmetric Boolean function $f$, $R(f) = O\left(Q(f)^2\right)$ as mentioned before. Researchers also studied the exact quantum query complexity for many instances of partial symmetric Boolean functions. For example,
Deutsch and Jozsa \cite{DJ92} studied the first partially symmetric Boolean function. 
Afterward, generalized Deutsch-Jozsa problems were studied in \cite{MJM15,QZ2016,QZ18}. 
He, Sun, Yang and Yuan \cite{HSY+18} established the asymptotically optimal bound for the exact quantum query complexity of distinguishing whether $|x|=k$ or $l$. Qiu and Zheng \cite{QZ2016,QZ20} studied the exact quantum query complexity of symmetric Boolean functions with degree 1 or 2. Additionally, several works \cite{BB14,FHKL16,LZ23} explored the connections between block sensitivity, fractional block sensitivity and degree for bounded functions.

In a similar work, Montanaro, Nishimura and Raymond \cite{MNR08} studied the unbounded error query complexity of Boolean functions in a scenario where
it is only required that the query algorithm succeeds with a probability strictly greater than
1/2. They proved quantum and classical query complexities are related by a constant factor for any (possibly partial) Boolean function. Similar results are also known in the communication complexity model \cite{INRY07A,INRY07}.
Compared to the result in \cite{MNR08}, we aim to analyze the relation between quantum/classical query complexity and bias more precisely. For instance, we show for any quantum algorithm computing $\fnk$ and $\kl$ using $T$ queries, there exist randomized algorithms using $\poly(T)$ queries that have the same bias as the quantum algorithm. Such a conclusion is not implied by \cite{MNR08} since the unbounded error model only cares about a strictly positive bias without a more accurate quantitative analysis.

\subsection{Organization}
The remainder of the paper is organized as follows. In Section \ref{sec:pre}, we review some notations, definitions and facts used in this paper. In Section \ref{sec:fnk}, the optimal success probabilities of quantum and classical randomized query algorithms to compute $\fnk$ are given when the number of queries is fixed. Furthermore, we study the exact quantum query complexity of $\fnk$.
In Section \ref{sec:kl}, we analyze the optimal success probability of quantum and classical query algorithms to compute $\kl$ given the number of queries. In Section \ref{sec:relationship_Q_R}, we consider the relation between the success probability of quantum and randomized algorithms to compute any total symmetric Boolean function for arbitrarily small bias.
In Section \ref{sec:relationship}, the relation of several complexity measures of partial symmetric Boolean functions is given. Finally, a conclusion is made in Section \ref{sec:conclusion}.

\section{Preliminary}\label{sec:pre}
This section gives some notations and definitions used in this paper. Moreover, we overview some useful facts and claims. Let $f:D \rightarrow \B$ be an $n$-bit Boolean function, where $D \subseteq \B^n$. If $D = \B^n$, $f$ is a total function. If $D \subset \B^n$, $f$ is a partial function. We say $f$ is a subfunction of $g: D_2 \rightarrow \B$, if $D \subseteq D_2$ and $f(x) = g(x)$ for any $x \in D$. We say $f$ is symmetric if $f(x)$ only depends on $|x|$, i.e., the number of $1$'s in $x$. 
If $k < 0$, we let $\binom{n}{k} = 0$. Let 
$[n] = \inbrace{1,...,n}$.
Every function $g:\inbrace{-1,1}^n \rightarrow \mathbb{R}$ can be uniquely expressed as
$g(x) = \sum_{S \subseteq [n]} \hat{g}(S) x_S$, 
where $x_S = \prod_{j \in S} x_j$ and $\hat{g}(S)$ is the Fourier coefficient of $g$  for any $S \subseteq [n]$.
\subsection{Query Models}
In the query model, a query algorithm uses a series of queries to the input to solve a specific problem. It works by making queries to an oracle, which is a function that provides answers to queries based on a specific problem. The algorithm then uses the responses to these queries to solve the problem. The goal of a query algorithm is to minimize the number of queries required to obtain a solution to the problem.  

In the classical query model, we can obtain $x_i$ for some $i$ by making one query. The classical query model can be discussed in two cases: the deterministic case and the randomized case. The deterministic query complexity of $f$, denoted by $D(f)$, is the minimum number of queries required by a classical deterministic algorithm to compute $f$ on the worst input. The randomized query complexity of $f$, denoted by $R_{\epsilon}(f)$, is the minimum number of queries required by a classical randomized algorithm to compute $f$ with error $\epsilon$ on the worst input. If the error is bounded, i.e., $\epsilon \le 1/3$, we abbreviate $R_{\epsilon}(f)$ to $R(f)$. Moreover, $R_0(f)$ is called the zero-error randomized query complexity of $f$.

In the quantum query model, a quantum query algorithm can be described as follows: it starts with a fixed state $|\psi_0\rangle$ and then performs the sequence of operations $U_0, O_x, U_1, \ldots, O_x, U_t$, where $U_i$'s are unitary operators not depend on $x$ and the query oracle $O_x$ is defined as 
\begin{equation}\label{eq:oracle}
O_x\ket{i}\ket{b} = \ket{i}\ket{x_i \oplus b}
\end{equation}
for any $i\in [n]$ and $b\in \B$. This leads to the final state $ |\psi_x\rangle=U_tO_xU_{t-1}\cdots U_1O_xU_0|\psi_0\rangle$. The output result is obtained by measuring the final state $ |\psi_x\rangle$. Similar to the classical case, the exact query complexity of $f$, denoted by $Q_E(f)$, is the minimum number of queries required by a quantum algorithm to compute $f$ exactly on the worst input. 
Such a quantum algorithm is called an exact quantum algorithm. The quantum query complexity of $f$, denoted by $Q_{\epsilon}(f)$, is the minimum number of queries required by a quantum algorithm to compute $f$ with $\epsilon$ on the worst input. If $\epsilon \le 1/3$, we abbreviate $Q_{\epsilon}(f)$ to $Q(f)$.

If an algorithm makes $T$ queries, we say the algorithm is a $T$-query algorithm. To characterize the optimal success probability of $T$-query algorithms formally, we give the following definitions.

\begin{definition}[Classical bias]
    Suppose $f:D\rightarrow\{0,1\}$ is an $n$-bit Boolean function, where $D \subseteq \B^n$. For an integer $T>0$, let the classical $T$-bias of $f$, denoted by $\delta_C(f,T)$, be the optimal success probability minus $1/2$ over all randomized algorithms using $T$ queries and all possible inputs. Namely,
    \begin{eqnarray*}
    &\hspace{-11cm}\delta_C(f,T)=\\
    &\max\{\delta:~\text{$\exists$ $T$-query randomized algorithm $\mathcal{A}$ s.t. $\forall x\in D$}~P(\mathcal{A}(x)=f(x))\geq1/2+\delta\}.
    \end{eqnarray*}
\end{definition}

\begin{definition}[Quantum bias]
    Suppose $f:D\rightarrow\{0,1\}$ is an $n$-bit Boolean function, where $D \subseteq \B^n$. For an integer $T>0$, let the quantum $T$-bias of $f$, denoted by $\delta_Q(f,T)$, be the optimal success probability minus $1/2$ over all quantum algorithms using $T$ quantum queries and all possible inputs.
     Namely,
    \begin{eqnarray*}
    &\hspace{-11cm}\delta_Q(f,T)=\\
    &\max\{\delta:~\text{$\exists$ $T$-query quantum algorithm $\mathcal{A}$ s.t. $\forall x\in D$}~P(\mathcal{A}(x)=f(x))\geq1/2+\delta\}.
    \end{eqnarray*}
\end{definition}

\subsection{Complexity measures}
In this section, we overview some definitions and facts about the complexity measures of Boolean functions.

\begin{definition}[Degree]\label{degree}
For an $n$-bit Boolean function $f:D \rightarrow \B$, where $D \subseteq \B^n$,
an $n$-variate polynomial $p:\mathbb{R}^n\rightarrow \mathbb{R}$ is called a real-valued multilinear polynomial representation of $f$ if $f(x)=p(x)$ for all $x\in D$ and $0\le p(x) \le 1$ for all $x\in \B^n$.
The \emph{degree} of $f$, denoted by $\deg(f)$, is defined as the minimum degree of all real-valued multilinear polynomial representations of $f$.
\end{definition}	
\begin{definition}[Approximate degree]
Suppose $f:D \rightarrow \B$ is an $n$-bit Boolean function, where $D \subseteq \B^n$. For $0\leq \varepsilon < 1/2$, we say a real multilinear polynomial $p$ approximates $f$ with error $\varepsilon$ if:
\begin{enumerate}
\item[(1)]  $|p(x)-f(x)|\leq \varepsilon$ for all $x\in D$;
\item[(2)]  $0\leq p(x)\leq 1$ for all $x\in\{0,1\}^n$.
\end{enumerate}
The approximate degree of $f$ with error $\epsilon$, denoted by $\widetilde{ \deg}_{\epsilon}(f)$, is the minimum degree among all real multilinear polynomials that approximate $f$ with error $\epsilon$. If $\epsilon \le 1/3$, we abbreviate $\widetilde{ \deg}_{\epsilon}(f)$ as $\widetilde{ \deg}(f)$. 
\end{definition}
\begin{definition}[Block sensitivity]
Suppose $f:D \rightarrow \B$ is an $n$-bit Boolean function, where $D \subseteq \B^n$.
    The block sensitivity $\bs(f,x)$ of $f$ on an input $x$ is defined as the maximum number $t$ such that there are $t$ pairwise disjoint subsets $B_1, \ldots , B_t$ of $[n]$ satisfying $x,x^{B_i}\in D$ and $f(x) \neq f\left(x^{B_i}\right)$, where $x^{B_i}$ denotes the string obtained by flipping the values of $x_i$ for all $i\in B$. Each $B_i$ is called a block. The block sensitivity $\bs(f)$ of $f$  is defined as $\max_{x \in D} \bs(f,x)$.
\end{definition}

\begin{definition}[Fractional block sensitivity]
Suppose $f:D \rightarrow \B$ is an $n$-bit Boolean function, where $D \subseteq \B^n$.
The fractional block sensitivity $\fbs(f,x)$ of $f$ on an input $x$ is the optimal value of the following linear program: 
\begin{align*}
\max \sum_{\genfrac{}{}{0pt}{}{B:x,x^B \in D,} {f(x) \neq f(x^B)}} w_B \hspace{1.5cm} 
\text{subject to: } &\forall i \in [n]: \sum_{B \ni i} w_B \leq 1, \\
&\forall B: 0 \leq w_B \leq 1.
\end{align*}
The fractional block sensitivity of $f$ is defined as $\fbs(f) = \max_{x \in D} \fbs(f,x)$. 
\end{definition}
\begin{definition}[Fractional certificate complexity]
Let $f:D \rightarrow \B$ be a Boolean function, where $D \subseteq \B^n$.  The fractional certificate complexity of $f$ on $x$, denoted by $\FC(f, x)$, is defined as
\begin{align*}
\min \sum_{i \in [n]} w_i \hspace{1.5cm} 
\text{subject to: } &\forall B \text{ s.t. } x,x^B \in D, f(x) \neq f(x^B): \sum_{i \in B} w_i \geq 1, \\
&\forall i \in [n]: 0 \leq w_i \leq 1.
\end{align*}
Then $\FC(f) = \max_{x \in D}\FC(f,x)$. 
\end{definition}
\begin{Fact}[\cite{BBC+01}]\label{beals}
If $f$ is a Boolean function, then $Q_E(f)\ge \deg(f)/2$ and $Q_{\epsilon}(f) \ge \adeg_{\epsilon}(f)/2$.
\end{Fact}
\begin{Fact}[\cite{GSS13,Tal13}]\label{bs_fbs_FC}
If $f$ is a Boolean function, then $\bs(f) \le \fbs(f) = \FC(f)$.
\end{Fact}
While \cite{BBC+01,GSS13,Tal13} focused on the total Boolean functions, it is not hard to verify that Facts \ref{beals} and \ref{bs_fbs_FC} also work for the partial case.

\subsection{Distance Measures}
Let $H\left(n,i,T\right)$ be the hypergeometric distribution sampling $T$ times from $x \in \B^n$ satisfying that $|x| = i$ without replacement. 
A binomial distribution with parameters $n,p$ is written as $B(n,p)$ and a Bernoulli distribution with parameter $p$ is written as $B(p)$. Given a probability distribution $p$, let $p^{\otimes T}$ denote the $T$-fold product distribution of $p$. Then the definition of total variation distance and the Hellinger distance are given as follows. 
\begin{definition}[The total variation distance]\label{def:TV}
Given two discrete probability distributions $p, q$ over probability
space $S$, the total variation distance $d_{TV}(p,q)$ between $p$ and $q$ is defined as:
\begin{equation*}
d_{TV}(p,q) = \frac{1}{2}\sum_{i \in S} |p(i)-q(i)|.
\end{equation*}
\end{definition}
\begin{definition}[The Hellinger distance]\label{def:Hell_dis}
The Hellinger distance between two discrete probability distributions $p, q$ over probability
space $S$ is defined as
    \begin{equation*}
    d_H^2(p,q) = \frac{1}{2}\sum_{i \in S} \left(\sqrt{p(i)}-\sqrt{q(i)}\right)^2.
    \end{equation*}
\end{definition}
\begin{Fact}[\cite{Lee2020}]\label{fact:dH_dTV_1}
For the total variation distance and the Hellinger distance, we have the following properties:
\begin{equation*}
\begin{aligned}
 d_H^2(p,q) &\le  d_{TV}(p,q) \le  \sqrt{2}d_H(p,q),\\
  d_{TV}(p^{\otimes T}, q^{\otimes T}) &\le T \cdot d_{TV}(p,q),    \\
  d_{TV}\left(p^{\otimes T},q ^{\otimes T}\right) &\ge d_H^2\left(p^{\otimes T},q^{\otimes T}\right) = 1-\left(1-d_H^2\left(p,q\right)\right)^T. \\
\end{aligned}
\end{equation*}
\end{Fact}
\begin{Fact}[Theorem 3.1 in 
\cite{holmes2004stein}]\label{fact:dtv_BH}
For any $0 \le k \le n$, $d_{TV}\left(B\left(T,\frac{k}{n}\right), H\left(n,k,T\right)\right) \le \frac{T-1}{n-1}.$
\end{Fact}
For the Hellinger distance between two Bernoulli distributions, we give the following claim: 
\begin{claim}\label{lemma:Hellinger_kl}
Let $p = B\left(\frac{k}{n}\right)$, $q = B\left(\frac{l}{n}\right)$ and $k < l$. Then
    $d_H^2 \inparen{p^{\otimes T}, q^{\otimes T}} = \Theta\inparen{\frac{\left(l-k\right)^2}{\left(n-k\right)l}}$.
\end{claim}

\begin{proof}
By Definition \ref{def:Hell_dis}, we have
    \begin{equation*}
d_H^2\left(p,q\right) =
\frac{1}{2}\inparen{ \inparen{\sqrt{\frac{k}{n}}-\sqrt{\frac{l}{n}}}^2+\inparen{\sqrt{\frac{n-k}{n}}-\sqrt{\frac{n-l}{n}}}^2}.
\end{equation*}
If $k \le n/2$, then \begin{equation}\label{inequality_kl_lower} 
\sqrt{\frac{l}{n}}-\sqrt{\frac{k}{n}} = \frac{l-k}{\sqrt{n}\left(\sqrt{l}+\sqrt{k}\right)} \ge \frac{l-k}{2\sqrt{2}\sqrt{\left(n-k\right)l}},
\end{equation}
since $\sqrt{n} \le \sqrt{2\left(n-k\right)}$ and $\sqrt{l}+\sqrt{k} \le 2\sqrt{l}$. Similarly, if $k > n/2$, then
\begin{equation*}
\frac{\sqrt{n-k}-\sqrt{n-l}}{\sqrt{n}} = \frac{l-k}{\sqrt{n}\left(\sqrt{n-l}+\sqrt{n-k}\right)} \ge 
\frac{l-k}{2\sqrt{2}\sqrt{\left(n-k\right)l}}
.
\end{equation*}
Thus, 
we have 
\begin{equation*}
d_H^2\left(p,q\right) \ge \frac{1}{16}\frac{\left(l-k\right)^2}{\left(n-k\right)l},
\end{equation*}
Moreover, since
\begin{equation}\label{inequality_kl_upper}
\sqrt{\frac{l}{n}}-\sqrt{\frac{k}{n}} = \frac{l-k}{\sqrt{n}\left(\sqrt{l}+\sqrt{k}\right)} \le 
\frac{l-k}{\sqrt{\left(n-k\right)l}},
\end{equation}
and
\begin{equation*}
\frac{\sqrt{n-k}-\sqrt{n-l}}{\sqrt{n}} = \frac{l-k}{\sqrt{n}\left(\sqrt{n-l}+\sqrt{n-k}\right)} \le \frac{l-k}{\sqrt{\left(n-k\right)l}},
\end{equation*}
we have
\begin{equation*}
d_H^2\left(p,q\right) \le \frac{\left(l-k\right)^2}{\left(n-k\right)l}.
\end{equation*}
\end{proof}
Then we present the following claim about distinguishing two discrete probability distributions.
\begin{claim}\label{lemma:onequery}
Given an unknown discrete distribution and promising that the distribution is either $p$ or $q$. 
Using one sample, one can determine whether the distribution is $p$ or $q$ with bias $d_{TV}\left(p,q\right)/4$.
\end{claim}
\begin{proof}
We assume the probability space is $S$, i.e., $\sum_{i \in S} p\left(i\right) = \sum_{i \in S} q\left(i\right) = 1$.
Let $S_0 = \inbrace{i \in S|p\left(i\right) \ge q\left(i\right)}$ and $S_1 = \inbrace{i \in S|p\left(i\right) < q\left(i\right)}$. 
Suppose $a = \sum_{i \in S_1} p\left(i\right)$ and $b = \sum_{i \in S_1} q\left(i\right)$. Then $d_{TV}\left(p,q\right) = b-a$. We give Algorithm \ref{al:one_query} to distinguish $p$ from $q$ using one sample.
\begin{algorithm}
\caption{One sample to distinguish distributions $p,q$.}
\label{al:one_query}
Sample an element $i \in S$ according to the distribution.

Suppose $a+b \ge 1$. If $i\in S_1$, we output $q$ with the probability $\frac{1}{a+b}$ and output $p$ with the probability $\frac{a+b-1}{a+b}$; if $i \in S_0$, then we output $p$.

Suppose $a+b < 1$. If $i \in S_1$, we output $q$; if $i \in S_0$, we output $q$ with the probability $\frac{1-a-b}{2-a-b}$ and output $p$ with the probability $\frac{1}{2-a-b}$.
\end{algorithm}
Then we analyze the success probability of Algorithm \ref{al:one_query} as follows:
\begin{enumerate}
    \item  $a+b \ge 1$.  
    If the distribution is $p$, then the success probability is $a\cdot\frac{a+b-1}{a+b}+1-a = \frac{b}{a+b}$; if the the distribution is $q$, then the success probability is $b\cdot\frac{1}{a+b} = \frac{b}{a+b}$. Thus, the bias of the algorithm is $\frac{b}{a+b}-\frac{1}{2} = \frac{b-a}{2\left(a+b\right)}\ge \frac{b-a}{4}$.
    \item $a+b < 1$. 
    If the distribution is $p$, then the success probability is $\left(1-a\right)\frac{1}{2-a-b}$; if the distribution is $q$, then the success probability is $b+\left(1-b\right)\frac{1-a-b}{2-a-b} = \frac{1-a}{2-a-b}$. Thus, the bias of the algorithm is $\frac{1-a}{2-a-b}-\frac{1}{2} = \frac{b-a}{2\left(2-a-b\right)} \ge \frac{b-a}{4}$. 
\end{enumerate}
\end{proof}
We note that there exists a similar result to distinguish $p$ from $q$ with bias $d_{TV}\left(p,q\right)/2$, such as Theorem 3.4 in \cite{ODonnel2015Lecture}. The difference is that \cite{ODonnel2015Lecture} considers the case that the prior probabilities that two probability distributions $p$ and $q$ occur are both $1/2$; the setting in our paper has no prior distribution, and thus we need to consider the worst case.
\subsection{Exact Quantum Query Complexity of $\kl$}
This section overviews some results about the exact quantum query complexity of $f_n^{k,l}$ \cite{HSY+18}. For every $(x,y)\in [0,1]^2$ and $x<y$, the upper-left region $\textsf{UL}(x,y)$ and the lower-right region $\textsf{LR}(x,y)$ are defined as follows:
$$                  \begin{array}{ll}
                        & \textsf{UL}(x,y) = \left\{(\kappa,\lambda)\in I^2|\lambda x\geq \kappa y ,(1-\kappa)(1-y)\geq(1-\lambda)(1-x),\kappa<\lambda\right\}; \\
                        & \textsf{LR}(x,y) =\left\{(\kappa,\lambda)\in I^2|\lambda x\leq \kappa y ,(1-\kappa)(1-y)\leq(1-\lambda)(1-x),\kappa<\lambda,(\kappa,\lambda)\neq (x,y)\right\}.
                     \end{array}
$$
For every set $S\subseteq [0,1]^2$, let $\textsf{UL}(S) = \bigcup_{(x,y)\in S}\textsf{UL}(x,y)$ and $\textsf{LR}(S) = \bigcup_{(x,y)\in S}\textsf{LR}(x,y)$. 
For every $d\in\mathbb{N}$, let 
\begin{multline*}
S_d  = \inbrace{\left( \frac{1}{2}{\left(1-\cos\left(\frac{\gamma\pi}{2d}\right)\right)},\frac{1}{2}\inbrace{\left(1-\cos\left(\frac{(\gamma+1)\pi}{2d}\right)\right)}   \right) \;\Big|\; \gamma\in\inbrace{0,1,\dots,2d-1}} \\
  \bigcup \inbrace{\left( \frac{1}{2}{\left(1-\cos\left(\frac{\gamma\pi}{2d-1}\right)\right)},\frac{1}{2}{\left(1-\cos\left(\frac{(\gamma+1)\pi}{2d-1}\right)\right)}   \right) \;\Big|\; \gamma\in\inbrace{1,\dots,2d-3}}.
\end{multline*}
Then the upper and lower bounds of the exact quantum query complexity of $\kl$ are as follows. 
\begin{Fact}[Theorem 1 in \cite{HSY+18}]\label{fact:padding}
For every $d\in\mathbb{N}$ and $0\leq k<l\leq n$ with $k,l,n\in\mathbb{N}$, let $\kappa=\frac{k}{n}$ and $\lambda=\frac{l}{n}$. If $(\kappa,\lambda)\in \mathsf{UL}(S_d)$, then $Q_E(f_n^{k,l})\leq d.$
\end{Fact}
\begin{Fact}[Theorem 2 in \cite{HSY+18}]\label{fact:lowerbound}
For every $d\in\mathbb{N}$ and $0\leq k<l\leq n$ with $k,l,n\in \mathbb{N}$, let $\kappa=\frac{k}{n}$ and $\lambda=\frac{l}{n}$. If $(\kappa,\lambda)\in \mathsf{LR}(S_d)$, then $Q_E(f_n^{k,l})\geq d+1$ for sufficiently large $n$.
\end{Fact}
\begin{Fact}[Corollary 1 in \cite{HSY+18}]\label{fact:kl_exact_upper_bound}
If $k,l,n\in\mathbb{N}$ and $0\le k< l\le n$, then $Q_E(f_n^{k,l})=O\left(\frac{\sqrt{(n-k)l}}{l-k}\right)$.
\end{Fact}
\subsection{Miscellaneous}
\begin{Fact}[Claim 2 in \cite{Ambainis07}\label{sin_inequality}]
    For any $0 \le \theta \le \pi/2$, we have $2\theta/\pi \le \sin \theta \le \theta$.
\end{Fact}
\begin{Fact}[Corollary 4.17 in \cite{Ben2021polynomial}]\label{fact:Ben2021}
Let $f$ be a (possibly partial) Boolean function. Then
\begin{equation*}
Q_{\epsilon}(f) = \Omega\left( \sqrt{(1-2\epsilon)\bs(f)}\right).
\end{equation*}
\end{Fact}
\begin{Fact}[Fact 4 in \cite{Paturi92}]\label{face:bmi}
Let $p:[-1,1]\rightarrow \mathbb{R}$ be a polynomial of degree $d$. Then for $x \in [-1,1]$, we have
\begin{equation*}
\left|\left(\frac{1}{d}\sqrt{1-x^2}+\frac{1}{d^2}\right)p'\left(x\right)\right| \le 2||p||,
\end{equation*}
where $||p|| = \sup\inbrace{|p\left(x\right)|:x \in [-1,1]}$.
\end{Fact}
\begin{Fact}[Theorem 12 in \cite{BHMT02}]\label{amplitude_estimation}
    Given an unknown input $x \in \B^n$ and a quantum query oracle to $x$. Suppose $a = |x|/n$,
    there exists a quantum approximate counting algorithm outputting an estimate $\Tilde{a}$ such that \begin{equation*}
|\Tilde{a}-a| \le 2\pi \frac{\sqrt{a(1-a)}}{t}+\frac{\pi^2}{t^2},
\end{equation*}
with probability at least $\frac{8}{\pi^2}$ using $t$ queries.
\end{Fact}
\begin{Fact}[Corollary 2.3 in \cite{Levenshtein95}]\label{lemma:krav}
For any $0 \le j \le T$, the Kravchuk polynomial is defined as
\begin{equation*}
K_j(t,T) = \sum_{i=0}^j\binom{t}{i}\binom{T-t}{j-i}(-1)^{i}.
\end{equation*}
Then for any $0 \le l,m \le T$, there exists the following orthogonality property:
\begin{equation*}
\sum_{t = 0}^T \binom{T}{t}K_l(t,T)K_m(t,T) = 2^T\binom{T}{l}\delta_{l,m},
\end{equation*}
where 
\begin{equation*}
\delta_{l,m} = 
\begin{cases}
1,&\text{ if }l=m,\\
0,&\text{ if }l \neq m.
\end{cases}
\end{equation*}
\end{Fact}
\begin{Fact}[Theorem 1 in \cite{Iyer2021tight}]\label{fact:fourier_growth}
If symmetric function $f:\inbrace{-1,1}^n \rightarrow [-1,1]$ has degree $d$, then
\begin{equation*}
\sum_{S\subseteq [n]:|S| = l}|\hat{f}(S)| \le \frac{d^l}{l!},
\end{equation*}
where $\hat{f}(S)$ is the Fourier coefficients of $f$ for any $S \subseteq [n]$.
\end{Fact}
\begin{Fact}[Lemma 2 in \cite{buhrman2002complexity}]\label{symmetrization}
If $p:\mathbb{R}^n \rightarrow \mathbb{R}$ is a multilinear polynomial, then there exists a single-variate polynomial $q: \mathbb{R} \rightarrow \mathbb{R}$, of degree at most the degree of $p$, such that $p^{\text{sym}}\inparen{x} = q\inparen{|x|}$ for all $x\in \inbrace{0,1}^n$, where
\begin{equation*}
    p^{\text{sym}}(x) = \frac{\sum_{\pi \in S_n}p\inparen{\pi\inparen{x}}}{n!},
\end{equation*}
where $S_n$ is the symmetric group and $\pi(x) = \inparen{x_{\pi(1)},...,x_{\pi(n)}}$ for any $\pi \in S_n$.
\end{Fact}
\begin{Fact}[Parseval’s identity]\label{lemma:parseval}
For a function $g:[-1,1]\rightarrow [-1,1]$, if $g(x) = \sum_{i = 0}^{T} a_i T_i(x)$ for any $x \in [-1,1]$, where $T_i$ is the Chebyshev polynomial of degree $i$, then
\begin{equation*}
\int_{-1}^1 \frac{1}{\sqrt{1-x^2}}\inparen{g(x)}^2 dx = \pi a_0^2  + \frac{\pi}{2}\sum_{i=1}^{T} a_i^2.
\end{equation*}
\end{Fact}
\begin{proof}
By Equations (4.10), (4.11a) and (4.11b) of \cite{mason2002chebyshev} (Page 84), we have
\begin{equation*}
\int_{-1}^{1} \frac{1}{\sqrt{1-x^2}}T_n(x) T_m(x) \, dx = \begin{cases} 0, &\mbox{if } n \neq m,\\ \pi, & \mbox{if } n = m = 0, \\ \frac{\pi}{2}, & \mbox{if } n = m \neq 0. \end{cases}
\end{equation*}
Thus,
\begin{equation*}
\begin{aligned}
    \int_{-1}^1 \frac{1}{\sqrt{1-x^2}}\inparen{g(x)}^2 dx &= \int_{-1}^1 \frac{1}{\sqrt{1-x^2}}\inparen{\sum_{i = 0}^ {T}a_i T_{i}(x)}^2 dx  \\
    &=  \pi a_0^2  + \frac{\pi}{2}\sum_{i=1}^{T} a_i^2.
\end{aligned}
\end{equation*}
\end{proof}
\begin{claim}\label{claim1e}
If $x > 1$, then $(1-\frac{1}{x})^{x-1} \ge \frac{1}{e}$.
\end{claim}
\begin{proof}
    Let $f(x)=(\frac{x}{x+1})^x = e^{x\ln{\frac{x}{x+1}}}$. Let $g(x) = x \ln{\frac{x}{x+1}}$. If $x > 0$, then $g'(x) = \ln{\frac{x}{x+1}}+(\frac{1}{x}-\frac{1}{x+1})x = \ln(1-\frac{1}{x+1})+\frac{1}{x+1} < 0$. Thus, if $x>0$, $f(x)$ is a decreasing function and $f(x) = (1-\frac{1}{x+1})^x \ge \lim\limits_{n \to \infty}(1-\frac{1}{x+1})^x = \frac{1}{e}$. Let $h(x) = f(x-1)$. Then $h(x) = (1-\frac{1}{x})^{x-1} \ge \frac{1}{e}$ for $x > 1$.
\end{proof}

\begin{claim}\label{claim_1epT}
If $0 < p < 1$ and $T$ is an integer satisfying that $1 \le T \le \frac{1}{p}$, then
\begin{equation*}
1-(1-p)^T \ge \frac{1}{e}pT.
\end{equation*}
\end{claim}
\begin{proof}
For $T=1$, it is trivial; for $1 < T \le \frac{1}{p}$, we have
\begin{equation*}
\begin{aligned}
1-(1-p)^T &= (1-p+p)^T-(1-p)^T \\
&\ge (1-p)^{T-1}pT\\
&\ge \inparen{1-p}^{\frac{1}{p}-1}pT \\
&\ge \frac{1}{e}pT,
\end{aligned}
\end{equation*}
where the last inequality comes from Claim \ref{claim1e}.
\end{proof}
\begin{claim}
\label{claim:sin_in}
If $T \ge 1$, $0 < \theta_2 < \theta_1 < \frac{\pi}{4T}$, then
$
\sin^2 \left(T\theta_1\right)-\sin^2 \left(T\theta_2\right) \ge \frac{1}{\pi}T^2\left(\sin^2 \theta_1-\sin^2 \theta_2\right).
$
\end{claim}
\begin{proof}
Let $f\left(\theta\right) = \sin^2 \left(T\theta\right)-\frac{1}{\pi}T^2\sin^2 \theta$.
Then 
\[
\begin{aligned}
f'\left(\theta\right) &= T\sin \left(2T\theta\right) - \frac{1}{\pi}T^2 \sin \left(2\theta\right) \\
&= T\left(\sin \left(2T\theta\right) - \frac{1}{\pi}T \sin \left(2\theta\right)\right).
\end{aligned}
\]
If $0 \le \theta \le \frac{\pi}{4T}$, we have 
\[\sin \left(2T \theta\right)-\frac{1}{\pi}T \sin \left(2\theta\right) \ge \frac{4T\theta}{\pi}-\frac{1}{\pi}T\cdot 2\theta \ge 0\] 
by Fact \ref{sin_inequality}, and thus $f'\left(\theta\right) \ge 0$. Since $f\left(0\right) =0$, we have $f\left(\theta\right) \ge 0$ for any $0 \le \theta \le \frac{\pi}{4T}$. Thus, for $0 < \theta_2 < \theta_1 < \frac{\pi}{4T}$, then 
\begin{equation*}
\sin^2 \left(T\theta_1\right) - \frac{1}{\pi}\sin^2\theta_1 \ge \sin^2 \left(T \theta_2\right) -\frac{1}{\pi}T^2\sin^2 \theta_2.
\end{equation*}
Thus, we have
\begin{equation*}
\sin^2 \left(T\theta_1\right) -\sin^2 \left(T\theta_2\right) \ge \frac{1}{\pi}T^2\left(\sin^2\theta_1-\sin^2 \theta_2\right).
\end{equation*}
\end{proof}

\section{The Optimal Success Probability of $T$-Query Algorithms to Compute $\fnk$}
\label{sec:fnk}
In this section, we show the optimal success probability of $T$-query quantum and classical randomized algorithms to compute $\fnk$ in Section \ref{subsec:fnk_quantum} and Section \ref{subsec:fnk_classical}. The definition of $\fnk:\B^n\rightarrow \B$ is as follows:
\begin{equation*}
\fnk(x) = 
\begin{cases}
0, & \text{ if }|x| \in \{0,n\},\\
1, & \text{ if }|x| \in \{k,n-k\}, \\
\end{cases}
\end{equation*}
where $1 \le k \le n/2$.

\subsection{The Quantum Case}\label{subsec:fnk_quantum}
In this section, we first show the upper bound on the bias of $T$-query quantum algorithms to compute $\fnk$ in Section \ref{subsubsec:fnk_quantumupperbound}. Then we give an optimal $T$-query algorithm to match the upper bound in Section \ref{subsubsec:fnk_quantumlowerbound}. For convenience, we let $\theta = 2\arcsin{\sqrt{k/n}}$. 

\subsubsection{The Upper Bound}\label{subsubsec:fnk_quantumupperbound}
If the success probability of an algorithm is $1-\epsilon$, then the bias of the algorithm is $\beta = 1-\epsilon-1/2=1/2-\epsilon$. By Fact \ref{fact:Ben2021}, we have $\Q_{\epsilon}(f) = \Omega(\sqrt{\beta\bs(f)})$. Thus, if the number of a quantum algorithm is no more than $T$ queries, then the bias $\beta$ of the algorithm is at most 
\begin{equation}\label{eq:bias_bs}
\beta  = O\left(\frac{T^2}{\bs(f)}\right).
\end{equation}
Next, we prove the following lemma:
\begin{lemma}\label{lemma:bs}
If $f$ is a partial symmetric Boolean function: $D \rightarrow \B$, where $D \subseteq \inbrace{0,1}^n$,  then 
\begin{equation*}
\bs(f) \ge \max_{a < b:f(a) \neq f(b)} \left\lfloor\frac{n}{2(b-a)}\right\rfloor.
\end{equation*} 
\end{lemma}

\begin{proof}
Suppose $f(a) \neq  f(b)$ and $a < b$. Let $x = \underbrace{0\cdots 0}_{n-a}\underbrace{1\cdots1}_{a}$. For $1 \le i \le \lfloor \frac{n-a}{b-a} \rfloor$, let 
\begin{equation*}
B_i = \inbrace{(i-1)(b-a)+1,(i-1)(b-a)+2,...,i(b-a)}.
\end{equation*}
Then $f(x) \neq f(x^{B_i})$.
Thus, $\bs(f) \ge \lfloor \frac{n-a}{b-a} \rfloor$. Using a similar argument, we have $\bs(f) \ge \lfloor \frac{b}{b-a} \rfloor$. Since $n-a+b \ge n$, it follows that $\bs(f) \ge \max\inbrace{\lfloor \frac{n-a}{b-a} \rfloor,\lfloor \frac{b}{b-a} \rfloor} \ge \lfloor\frac{n}{2(b-a)}\rfloor$. 
\end{proof}
By Lemma \ref{lemma:bs}, for the function $\fnk$, we have $\bs(\fnk) = \Omega\inparen{n/k}$.
By \Cref{eq:bias_bs}, we have the following lemma:
\begin{lemma}\label{lemma:fnk_quantum_upperbound}
For the function $\kl$ and $T>0$, the quantum $T$-bias satisfies that
$\delta_Q\inparen{\fnk,T} \le O\inparen{\frac{k}{n}\cdot T^2}$.
\end{lemma}
Furthermore, we aim to give the range of $T$ such that the optimal success probability of $T$-query quantum algorithms is less than 1. Equivalently, we aim to prove the lower bound of the exact quantum query complexity of $\fnk$. We give the following lemma:
\begin{lemma}\label{lemma:fnk_exact_lower_bound}
The exact quantum query complexity of $\fnk$ is at least $\lceil \frac{\pi}{2\theta}\rceil$.
\end{lemma}
\begin{proof}
Since $f_n^{0,k}$ is a subfunction of $\fnk$, which is defined as
\begin{equation*}
f_n^{0,k}(x) = 
\begin{cases}
0 & \text{ if }|x| = 0, \\
1 & \text{ if }|x| = k,
\end{cases}
\end{equation*}
we have $Q_E(\fnk) \ge Q_E(f_n^{0,k})$. Fact \ref{fact:lowerbound} implies that if 
\begin{equation}\label{eq:fn0k_upper_bound}
\frac{k}{n} \le \frac{1}{2}\inparen{1-\cos \frac{\pi}{2d}},
\end{equation}
then $Q_E(f_n^{0,k}) \ge d+1$. Let $d = \lceil \frac{\pi}{2\theta}\rceil-1$. Then $d$ satisfies Equation (\ref{eq:fn0k_upper_bound}) and thus $ Q_E(f_n^{0,k}) \ge \lceil \frac{\pi}{2\theta}\rceil$. Thus, we have
$Q_E(\fnk) \ge \lceil \frac{\pi}{2\theta}\rceil$.
\end{proof}

\subsubsection{The Lower Bound}\label{subsubsec:fnk_quantumlowerbound}
This section gives an optimal $T$-query quantum algorithm to match the above upper bound on
the success probability. The algorithm is a variant of amplitude amplification algorithm \cite{BHMT02}. 
Let 
\begin{equation*}
\ket{\psi} = \frac{1}{\sqrt{n}}\sum_{i\in [n]} \ket{i}\ket{-}.
\end{equation*}
Then 
\begin{equation*}
O_{x} \ket{\psi} = \frac{1}{\sqrt{n}}\sum_{i\in [n]} (-1)^{x_i}\ket{i}\ket{-}
\end{equation*}
Thus, 
\begin{equation*}
\langle \psi|O_x | \psi\rangle  = \frac{1}{n}(n-2|x|) = 1- \frac{2|x|}{n}.    
\end{equation*}
As a result, there exists a state $\ket{\psi^{\bot}}$ such that $\langle \psi | \psi^{\bot} \rangle = 0$ and 
\begin{equation}\label{eq:Ox_psi_fnk}
O_x \ket{\psi} = \cos \theta_x \ket{\psi}+\sin \theta_x \ket{\psi^{\bot}},
\end{equation}
where $\cos \theta_x = 1-\frac{2|x|}{n}$, i.e., $\theta_x = 2\arcsin{\sqrt{\frac{|x|}{n}}}$. 
Two reflection operations are defined as follows:

\begin{equation}\label{eq:S0S1_fnk}
\begin{aligned}
S_{0} = 2\ket{\psi}\bra{\psi}-I,\quad\mbox{and}\quad S_{1} = 2O_x\ket{\psi}\bra{\psi}O_x-I.
\end{aligned}
\end{equation}
Then we have 
\begin{equation}\label{eq:S1S0_psi_fnk}
\begin{aligned}
& S_1S_0\ket{\psi}& & =\cos 2\theta_x\ket{\psi}+\sin 2\theta_x\ket{\psi^{\bot}},\\
& S_1S_0\ket{\psi^{\bot}}& & = -\sin2\theta_x\ket{\psi}+\cos2\theta_x\ket{\psi^{\bot}}.
\end{aligned}
\end{equation}
Since $O_xS_{1} O_x = 2\ket{\psi}\bra{\psi}-I,$
we have
\begin{equation*}
S_{1} = O_x(2\ket{\psi}\bra{\psi}-I)O_x = O_x
S_{0}O_x.
\end{equation*}
Let 
\begin{equation}\label{eq:pro_fnk}
P_0 = \ket{\psi}\bra{\psi}, P_1 = I - \ket{\psi}\bra{\psi}.
\end{equation}
When the number of queries $T$ satisfies $T \le \sqrt{n/k}$, we give Algorithm \ref{al:T_query_fnk} to compute $\fnk$.
\begin{algorithm}
\caption{A $T$-query quantum algorithm to compute $\fnk$.}
\label{al:T_query_fnk}
Prepare the initial state $\ket{\psi} = \frac{1}{\sqrt{n}}\sum_{i\in [n]} \ket{i}\ket{-}$. 

If $T = 2m$ for some integer $m$, then perform the operation $(S_1S_0)^{m}$ to $\ket{\psi}$; if $T = 2m+1$ for some integer $m$, then perform the operation $(S_1S_0)^{m}O_x$ to $\ket{\psi}$, where $S_0,S_1$ are defined as \Cref{eq:S0S1_fnk}.

Perform the project measurement $\inbrace{P_0,P_1}$ given in \Cref{eq:pro_fnk} to the final state. 
If the measurement result is 1, we output 1; if the measurement result is 0, we output 0 with the probability $\frac{1}{2-\sin^2 T\theta}$ and output 1 with the probability $1-\frac{1}{2-\sin^2 T\theta}$, where $\theta = 2\arcsin{\sqrt{\frac{k}{n}}}$. 
\end{algorithm}
By \Cref{eq:Ox_psi_fnk,eq:S1S0_psi_fnk}, whichever $T$ is even or odd, after running Step 2 of Algorithm \ref{al:T_query_fnk}, the final state is 
\begin{equation*}
\cos T \theta_x \ket{\psi}+\sin T \theta_x\ket{\psi^{\bot}}. 
\end{equation*}
We analyze the success probability of the algorithm as follows: 
If $|x| = 0$ or $n$, then $\theta_x = 0$ or $\pi$. Thus, we have $\cos^2 T\theta_x = 1$ and the measurement result is 0. As a result, the success probability of the algorithm is  $\frac{1}{2-\sin^2 T \theta}$. If $|x| = k$ or $n-k$, then $\theta_x = \theta$ or $\pi-\theta$. Thus, we have $\cos^2 T\theta_x = \cos^2 T\theta$, and the success probability of the algorithm is 
\begin{equation*}
\inparen{1-\frac{1}{2-\sin^2 T \theta}}\cos^2 T\theta  + 1-\cos^2 T\theta = \frac{1}{2-\sin^2 T \theta}.
\end{equation*}
Thus, the bias of the algorithm is 
\begin{equation*}
\frac{1}{2-\sin^2 T\theta}-\frac{1}{2} \ge \frac{\sin^2 T \theta}{4} \ge \frac{1}{\pi^2}T^2\theta^2 \ge \frac{4}{\pi^2}\frac{k}{n}\cdot T^2.
\end{equation*}
If $T > \sqrt{n/k}$, we run Algorithm \ref{al:T_query_fnk} with $\sqrt{n/k}$ queries and let other queries be dummies. Then the success probability of the algorithm is at least $\frac{4}{\pi^2}$.
Thus we have the following lemma:
\begin{lemma} \label{lemma:fnk_quantum_lowerbound}
For the function $\fnk$ and $T>0$, the quantum $T$-bias satisfies that
\begin{equation*}
\delta_Q\inparen{\fnk, T}= 
\begin{cases}
\Omega(\fnkquanbias),&  \text{ if } T \le  \sqrt{n/k},\\
\Omega\inparen{1},&\text{ if } T > \sqrt{n/k},
\end{cases}
\end{equation*}
\end{lemma}
Combining Lemma \ref{lemma:fnk_quantum_upperbound} and Lemma \ref{lemma:fnk_quantum_lowerbound}, we have the following theorem:
\begin{theorem}\label{th:fnk_quantum}
For the function $\fnk$ and $T>0$, the quantum $T$-bias is
\begin{equation*}
\delta_Q\inparen{\fnk, T}= 
\begin{cases}
\Theta(\fnkquanbias),&  \text{ if } T \le  \sqrt{n/k},\\
\Theta\inparen{1},&\text{ if } T > \sqrt{n/k},
\end{cases}
\end{equation*}
\end{theorem}
Furthermore, we aim to show if $T \ge \lceil \frac{\pi}{2\theta} \rceil+2$, the optimal success probability of $T$-query quantum algorithm is 1. We give the following lemma:
\begin{lemma}\label{lemma:fnk_exact_upper_bound}
    The exact quantum query complexity of $\fnk$ is at most $\lceil \frac{\pi}{2\theta} \rceil+2$.
\end{lemma}
\begin{proof}
Fact \ref{fact:padding} implies if 
\begin{equation}\label{eq:fn0k_lower_bound}
\frac{k}{n} \ge \frac{1}{2}\inparen{1-\cos \frac{\pi}{2d}},
\end{equation}
then $Q_E(f_n^{0,k}) \le d$. Let $d = \lceil \frac{\pi}{2\theta}\rceil$. Then $d$ satisfies Equation (\ref{eq:fn0k_lower_bound}) and thus $ Q_E(f_n^{0,k}) \le \lceil \frac{\pi}{2\theta}\rceil$. As a result, promising either $|x| = 0$ or $|x| =k$, there exists an exact quantum algorithm $\mathcal{A}_{0,k}$ distinguish $|x| = 0$ from $|x| = k$ with $\lceil \frac{\pi}{2\theta} \rceil$ queries, where $\theta = 2\arcsin{\sqrt{k/n}}$. 
Note that when $|x| \notin \inbrace{0,k}$, although $\mathcal{A}_{0,k}$ does not promise to output the correct result with certainty, it can give some information about $x$. Specifically, if the output of $\mathcal{A}_{0,k}$ is $|x| = 0$, then it implies $|x| \neq k$ and vice versa. 

We give an exact quantum algorithm to compute $f_n^k(x)$ for $|x| \in \inbrace{0,k,n-k,n}$ by utilizing the above property of $\mathcal{A}_{0,k}$: We run $\mathcal{A}_{0,n-k}$ to $x$ and $\overline{x}$ respectively, where $\overline{x}$ is the string obtained by flipping the values of $x_i$ for all $i\in [n]$. Then we discuss the following four cases. 
\begin{enumerate}
    \item The outputs are $|x| = 0$, $|\overline{x}| = 0$. It implies $|x| \neq n-k$ and $|\overline{x}| \neq n-k$. Thus $|x| \in \inbrace{0,n}$ and $\fnk(x) = 0$.
    \item The outputs are $|x| =n-k$, $|\overline{x}| = n-k$. It implies $|x| \neq 0$ and $|\overline{x}| \neq 0$. Thus $|x| \in \inbrace{k,n-k}$ and $\fnk(x) = 1$.
    \item \label{item_exact_nk}
    The outputs are $|x| = 0$, $|\overline{x}| = n-k$. It implies $|x| \neq n-k$ and $|\overline{x}| \neq 0$, and thus $|x| \in \inbrace{0,k}$. So we distinguish $|x| = 0$ from $|x| = k$ by running $\mathcal{A}_{0,k}$ to $x$. If $|x| = 0$, then $\fnk(x) = 0$; if $|x| = k$, then $\fnk(x) = 1$.
    \item The outputs are $|x| = n-k$, $|\overline{x}| = 0$. The case is similar to Case \ref{item_exact_nk}. 
\end{enumerate}
Since $k \le n/2$, the number of queries in the algorithm is 
\begin{equation*}
2\left\lceil \frac{\pi}{2\cdot 2\arcsin{\sqrt{\frac{n-k}{n}}}} \right\rceil +\left\lceil \frac{\pi}{2\cdot 2\arcsin{\sqrt{\frac{k}{n}}}}\right\rceil = 2+ \left\lceil \frac{\pi}{2\theta}\right\rceil.
\end{equation*}
Thus, we have
$Q_E(\fnk) \le \left\lceil \frac{\pi}{2\theta}\right\rceil+2$.
\end{proof}

\subsection{The Classical Case}\label{subsec:fnk_classical}
This section consists of Section \ref{subsubsec:fnk_classicalupperbound} and \ref{subsubsec:fnk_classicallowerbound}, which show the upper and lower bounds on the success probability of $T$-query classical randomized algorithms, respectively.
\subsubsection{The Upper Bound}\label{subsubsec:fnk_classicalupperbound}

First, we give the following fact. 
\begin{Fact}[Merging Corollary 1 with Proposition 2 in \cite{Koiran2006on}]\label{proposition:classical_upperbound}
For a partial symmetric Boolean function $f:D \rightarrow \B$, where $D \subseteq \B^n$, let $X = \inbrace{x|f(x) = 0}$ and $Y = \inbrace{y|f(y) = 1}$. Let $I \in [n]^T$ be a list of queries and $B \in \B^T$ be a set of possible answers. Let $x_I = x_{I_1}\cdots x_{I_T}$. 
Let $P_I^X(B)$ be the proportion of $x \in X$ satisfying the condition $x_I \in B$ and $P_I^Y(B)$ be the proportion of $y \in Y$ satisfying the condition $y_I \in B$. Let 
\begin{equation}\label{eq:alpha}
\alpha = \max_{I \in [n]^T, B \in \B^T}|P_I^X(B)-P_I^Y(B)| = \frac{1}{2}\sum_{I \in [n]^T, B \in \B^T}|P_I^X(B)-P_I^Y(B)|.
\end{equation}
Then $\delta_C\inparen{f, T} \le \alpha/2$.
\end{Fact}
Then we prove the following lemma using \Cref{proposition:classical_upperbound}:
\begin{lemma}\label{lemma:fnk_classical_upperbound}
    For the function $\fnk$, the classical $T$-bias  satisfies:
    \begin{equation*}
      \delta_C\inparen{\fnk, T} = 
        \begin{cases}
         0,& \text{ if }T=1,\\
         O\inparen{\frac{k}{n}\cdot T},&\text{ if } 2 \le T \le n/k.
        \end{cases}
    \end{equation*}
\end{lemma}
\begin{proof}
We assume $T \le n/k$.
When $|B| = 0$ or $T$, $P_I^X(B)= 1/2$; when $0 < |B| < T$, $P_I^X(B) = 0$. On the other hand, for any $B \in \B^T$,
\begin{equation*}
P_I^Y(B) = \frac{\binom{n-T}{k-|B|}+\binom{n-T}{k-T+|B|}}{\binom{n}{k}+\binom{n}{n-k}}.
\end{equation*}
Then 
\begin{equation}\label{eq:gamma_alpha}
      \delta_C\inparen{\fnk, T} \le \frac{\alpha}{2},
\end{equation}
where
\begin{equation*}
\begin{aligned}
\alpha &= \max_{I \in [n]^T, B \in \B^T}|P_I^X(B)-P_I^Y(B)| \\
&= \sum_{|B| = 0 \text{ or } T}\inparen{P_I^X(B)-P_I^Y(B)} \\
&= 1-\frac{\binom{n-T}{k}+\binom{n-T}{k-T}}{\binom{n}{k}}. \\
\end{aligned}
\end{equation*}
If $T = 1$, then 
\begin{equation}\label{eq:alpha1}
\alpha = 1-\frac{\binom{n-1}{k}+\binom{n-1}{k-1}}{\binom{n}{k}} = 0;
\end{equation}
if $2 \le T \le n/k$, then 
\begin{equation}\label{eq:alpha2}
\begin{aligned}
\alpha &= 1-\frac{\binom{n-T}{k}+\binom{n-T}{k-T}}{\binom{n}{k}} \\
&\le 1-\frac{\binom{n-T}{k}}{\binom{n}{k}} \\
&\le 1-\inparen{1-\frac{T}{n-k+1}}^k \\
&\le \frac{kT}{n-k+1}\\
&\le \frac{2kT}{n}.
\end{aligned}
\end{equation}
Combining \Cref{eq:gamma_alpha,eq:alpha1,eq:alpha2}, Lemma \ref{lemma:fnk_classical_upperbound} is proven.
\end{proof}

\subsubsection{The Lower Bound}
\label{subsubsec:fnk_classicallowerbound}
We give the following lemma to match the upper bound above:
\begin{lemma}\label{lemma:fnk_classical_lowerbound}
    For the function $\fnk$ and $2 \le T \le n/k$, the classical $T$-bias satisfies:
    \begin{equation*}
      \delta_C\inparen{\fnk, T} = 
    \Omega\inparen{\frac{k}{n}\cdot T}.
    \end{equation*}
\end{lemma}
\begin{proof}
For $2 \le T \le n/k$, we give the following $T$-query algorithm: We make $T$ queries uniformly and independently. If the query result is all 1's or 0's, we output 0 with the probability $1-p$ and output 1 with the probability $p$, where 
\begin{equation*}
p = \frac{\inparen{\frac{k}{n}}^T+\inparen{\frac{n-k}{n}}^T}{\inparen{\frac{k}{n}}^T+\inparen{\frac{n-k}{n}}^T+1}.
\end{equation*}
Otherwise, we output 1.

Next, we analyze the bias of the algorithm. If $|x| = 0$ or $n$, the success probability of the algorithm is $1-p$; if $|x| = k$ or $n-k$, the success probability of the algorithm is 
\begin{equation*}
p\inparen{\inparen{\frac{k}{n}}^T+\inparen{\frac{n-k}{n}}^T}+\inparen{1-\inparen{\inparen{\frac{k}{n}}^T+\inparen{\frac{n-k}{n}}^T}} = 1-p.
\end{equation*}
Thus, the bias of the algorithm is 
\begin{equation*}
\begin{aligned}
1-p-\frac{1}{2} &= \frac{1-\inparen{\frac{k}{n}}^T-\inparen{\frac{n-k}{n}}^T}{2\inparen{\inparen{\frac{k}{n}}^T+\inparen{\frac{n-k}{n}}^T+1}}\\
&\ge \frac{1}{4}\inparen{1-\inparen{\frac{n-k}{n}}^T-\inparen{\frac{k}{n}}^T}\\
&\ge \frac{1}{4}\cdot\inparen{\frac{1}{e}\frac{kT}{n}-\inparen{\frac{k}{n}}^T} \\
&\ge \frac{1}{4}\cdot\inparen{\frac{1}{e}\frac{kT}{n}-\frac{1}{2}\frac{k}{n}} \\
&= \frac{1}{4}\cdot\frac{k}{n}\inparen{\frac{T}{e}-\frac{1}{2}},\\
\end{aligned}
\end{equation*}
where the second inequality comes from Claim \ref{claim_1epT}.
\end{proof}
\Cref{lemma:fnk_classical_lowerbound} implies if $T > n/k$, then $\delta_C\inparen{\fnk, T} = \Omega(1)$.
Combining Lemmas \ref{lemma:fnk_classical_upperbound} and \ref{lemma:fnk_classical_lowerbound}, we have the following theorem:
\begin{theorem}\label{th:fnk_classical}
For the function $\fnk$ and $T>0$, the classical $T$-bias is 
\begin{equation*}
\delta_C\inparen{\fnk, T} = 
\begin{cases}
0, & \text{ if }T=1, \\
\Theta(\fnkclassbias), & \text{ if }2 \le T \le n/k, \\
\Theta(1), & \text{ if }T > n/k.
\end{cases}
\end{equation*}
\end{theorem}
Finally, we give 
Claim \ref{claim:d_fnk}, 
which implies the zero-error randomized query complexity of $\fnk$ is $n-k+1$.
\begin{claim}\label{claim:d_fnk}
The optimal success probability of $T$-query classical randomized algorithms is 1 if and only if $T \ge n-k+1$.
\end{claim}

\begin{proof}
On the one hand, by  \Cref{proposition:classical_upperbound} and \Cref{eq:alpha}, if $T \le n-k$, 
the optimal success algorithm of $T$-query algorithms is at most 
\begin{equation*}
\frac{1}{2}+\frac{1}{2}\inparen{1-\frac{\binom{n-T}{k}+\binom{n-T}{k-T}}{\binom{n}{k}}} < 1.
\end{equation*}
On the other hand, we give the following algorithm using $n-k+1$ queries: Query $n-k+1$ distinct bits randomly. i) If the results are all 1's or all 0's, then $|x| \ge n-k+1$ or $|x| \le k-1$. Since $|x| \in \inbrace{0,k,n-k,n}$, we have $|x| = 0$ or $|x| = n$. Thus we output 0. ii) If the results include both 0 and 1, then we output 1. 
\end{proof}
Finally, \Cref{th:fnk1} consists of \Cref{th:fnk_quantum,th:fnk_classical};
\Cref{cor:fnk} is a direct corollary of \Cref{th:fnk1};
\Cref{th:fnk2} consists of 
Claim \ref{claim:d_fnk}, Lemma \ref{lemma:fnk_exact_lower_bound} and Lemma \ref{lemma:fnk_exact_upper_bound}.
We recall \Cref{th:fnk1}, \Cref{cor:fnk} and \Cref{th:fnk2} as follows:
\optfnk* 
\corfnk*
\exactfnk*

\section{The Optimal Success Probability of $T$-Query Algorithms to Compute $\kl$}
\label{sec:kl}
In this section, we analyze the optimal success probability of quantum and classical query algorithms to compute $\kl$ when the number of queries is fixed, where $k < l$ and
\begin{equation*}
\kl(x) = 
\begin{cases}
0, & \text{ if }|x| = k, \\
1, & \text{ if }|x| = l.
\end{cases}
\end{equation*}

\subsection{The Quantum Case}
\subsubsection{The Upper Bound}\label{kl_quantum_upper_bound}
We first give the lower bound of the approximate degree of $f_n^{k,l}$ for any $0 < \epsilon < 1/2$ by \Cref{lemma:deg_lowerbound}. The proof idea is inspired by \cite{Paturi92}.
\begin{restatable}{lemma}{deglowerbound}
\label{lemma:deg_lowerbound}
$\adeg_{\epsilon}\left(\kl\right) = \Omega\inparen{\max\inbrace{\frac{\beta\sqrt{\left(n-k\right)l}}{l-k}
, \sqrt{\frac{\beta n}{l-k}}}}$, where $\beta = 1/2-\epsilon$.
\end{restatable}
\begin{proof}
By Fact \ref{face:bmi}, we have 
\begin{equation}\label{eq:deg_lower_bound}
\deg\left(p\right) \ge \max\inbrace{\frac{\sqrt{1-x^2}|p'\left(x\right)|}{2||p||},\sqrt{\frac{|p'\left(x\right)|}{2||p||}}}.
\end{equation}
Let $g:\{0,1\}^n \rightarrow [0,1]$ be a function to approximate $\kl$ with bias $\beta$. 
we let $g:[0,n] \rightarrow [0,1]$  be the corresponding single-variate version of polynomial such that $g\left(k\right) \le 1/2-\beta$, $g\left(l\right) \ge 1/2+\beta$ and $g\left(i\right) \in [0,1]$ for any $i \in \inbrace{0,...,n}$. Let $p:[-1,1] \rightarrow [-1,1]$ defined as 
$p\left(\frac{2x}{n}-1\right) = g\left(x\right)$ for any $x \in [0,n]$. Then $p\left(\frac{2k}{n}-1\right) \le 1/2-\beta$, $p\left(\frac{2l}{n}-1\right)\ge 1/2+\beta$. For the interval $[\frac{2k}{n}-1,\frac{2l}{n}-1]$, if $||p|| \le 1$, by the mean value theorem, there exists $x \in [\frac{2k}{n}-1,\frac{2l}{n}-1]$ such that 
\begin{equation*}
p'\left(x\right) \ge \frac{\left(\frac{1}{2}+\beta\right)-\left(\frac{1}{2}-\beta\right)}{\left(\frac{2l}{n}-1\right)-\left(\frac{2k}{n}-1\right)} = \frac{n\beta}{l-k},
\end{equation*}
so we have
$\frac{|p'\left(x\right)|}{||p||} \ge \frac{n\beta}{l-k}$; if $||p|| \ge 1$, then there exists $x$ such that 
\begin{equation*}
p'\left(x\right) = \frac{||p||-\left(\frac{1}{2}-\beta\right)}{\frac{2\left(l-k\right)}{n}} \ge \frac{\frac{||p||}{2}}{\frac{2\left(l-k\right)}{n}} = \frac{||p||n}{4\left(l-k\right)} \ge  \frac{||p||n\beta}{4\left(l-k\right)}.
\end{equation*} 
Thus, for the interval $[\frac{2k}{n}-1,\frac{2l}{n}-1]$, there always exists $x$ such that
\begin{equation}\label{eq:deg_lower_bound_2}
\frac{|p'\left(x\right)|}{||p||} \ge \frac{n\beta}{4\left(l-k\right)}.
\end{equation}
Since $x \in [\frac{2k}{n}-1, \frac{2l}{n}-1]$, we have 
\begin{equation}\label{eq:inequ_sqrt_x}
\begin{aligned}
\sqrt{1-x^2} &\ge \min\inbrace{2\sqrt{\frac{k}{n}\left(1-\frac{k}{n}\right)}, 2\sqrt{\frac{l}{n}\left(1-\frac{l}{n}\right)}} \\
&\ge 2\sqrt{\min\inbrace{\frac{n-k}{n},\frac{l}{n}}}.
\end{aligned}
\end{equation}
Let $d_1 = \frac{\sqrt{1-x^2}p^{\prime}\left(x\right)}{2||p||}$. 
By \Cref{eq:deg_lower_bound_2,eq:inequ_sqrt_x}, we have
\begin{equation}\label{eq:d1}
\begin{aligned}
d_1 &=\Omega\inparen{\frac{n\beta\min\inbrace{\sqrt{n-k},\sqrt{l}}}{l-k}} \\
&=\Omega\inparen{\frac{\beta\max\inbrace{\sqrt{n-k},\sqrt{l}}\min\inbrace{\sqrt{n-k},\sqrt{l}}}{l-k}} \\
&= \Omega\inparen{\frac{\beta\sqrt{\left(n-k\right)l}}{l-k}}.  
\end{aligned}
\end{equation}
Combining \Cref{eq:deg_lower_bound,eq:deg_lower_bound_2,eq:d1}, Lemma \ref{lemma:deg_lowerbound} is proved.
\end{proof}
Then we give the following lemma:
\begin{lemma}\label{lemma:upperboundKL}
For the function $\kl$ and $T>0$, the quantum $T$-bias satisfies that 
\[
\delta_Q\inparen{\kl, T}= O\inparen{\min \inbrace{\frac{l-k}{\sqrt{\left(n-k\right)l}}\cdot T, \frac{l-k}{n} \cdot T^2}}.
\]
\end{lemma}
\begin{proof}
By \Cref{beals}, if a quantum algorithm uses $T$ queries succeed with probability at least $1-\epsilon = 1/2+\beta$, then 
\begin{equation}\label{eq:upperboundKL}
T \ge \frac{1}{2}\adeg_{\epsilon}\left(\kl\right) = \Omega\inparen{\max\inbrace{\frac{\beta\sqrt{\left(n-k\right)l}}{l-k}
, \sqrt{\frac{\beta n}{l-k}}}}.
\end{equation}
By \Cref{eq:upperboundKL}, we have
$\beta = O\inparen{\min \inbrace{\frac{l-k}{\sqrt{\left(n-k\right)l}}\cdot T, \frac{l-k}{n} \cdot T^2}}$.   
\end{proof}

\subsubsection{The Lower Bound}
In this section, we will prove the following lemma:
\begin{lemma}\label{lemma:lowerboundKL}
For the function $\kl$ and $T>0$, the quantum $T$-bias satisfies that
\begin{equation*}
\delta_Q\inparen{\kl, T}= 
\begin{cases}
\Omega\inparen{\min \inbrace{\frac{l-k}{\sqrt{(n-k)l}}\cdot T, \frac{l-k}{n} \cdot T^2}}, & \text{ if }T= O\inparen{\frac{\sqrt{(n-k)l}}{l-k}}, \\
\Omega(1), & \text{ if }T = \Omega\inparen{\frac{\sqrt{(n-k)l}}{l-k}}.
\end{cases}
\end{equation*}
\end{lemma}
\begin{proof}

If $l-k = \Omega\left(n\right)$, then $\frac{\sqrt{\left(n-k\right)l}}{l-k} = \Omega\left(1\right)$ and $\frac{l-k}{n} = \Omega\left(1\right)$. Moreover, Lemma \ref{lemma:onequery} implies there exists a one-query classical algorithm such that the bias of the algorithm is $\Omega\left(1\right)$. Thus, there also exists a one-query quantum algorithm such that the bias of the algorithm is $\Omega\left(1\right)$. As a result, \Cref{lemma:lowerboundKL} is correct when $l-k = \Omega(n)$. 
Next, we only need to consider the case $l-k = o\left(n\right)$. 
We first assume $k \le n/2$. 
Then we further assume $k \le l \le n/4$, since if $l > n/4$, we can pad $3n$ 0's to the input and obtain a new $n'$-bit Boolean input string, where $n' = 4n$ and $l \le n'/4$. Since $n' = 4n$, the complexity analysis only has a difference of constant factor. Then the upper bound in Lemma \ref{lemma:upperboundKL} satisfies 
\begin{equation*}
    O\inparen{\min \inbrace{\frac{l-k}{\sqrt{\left(n-k\right)l}}\cdot T, \frac{l-k}{n} \cdot T^2}} = 
    \begin{cases}
        O\inparen{\frac{l-k}{n} \cdot T^2}, &\text{ if }T = O\inparen{\sqrt{\frac{n}{l}}},\\
        O\inparen{\frac{l-k}{\sqrt{\left(n-k\right)l}}\cdot T}, &\text{ if }T = \Omega\inparen{\sqrt{\frac{n}{l}}}.
    \end{cases}
\end{equation*}
Let $\theta_l = \arcsin{\sqrt{l/n}}$ and $\theta_k = \arcsin{\sqrt{k/n}}$. Then $\theta_k \le \theta_l \le \pi/6$. We discuss the following cases according to the range of the number of queries $T$.

i) If $T \le \frac{1}{4}\cdot \sqrt{\frac{n}{l}}$, then 
$\left(2T-1\right)\theta_l \le 2T \cdot \frac{\pi}{2}
\sqrt{\frac{l}{n}} \le \pi/4$. We give Algorithm \ref{al:T_queryKL1} to compute $\kl$.
\begin{algorithm}
\caption{A $T$-query quantum algorithm to compute $\kl$.}
\label{al:T_queryKL1}
\KwIn{$x\in\B^n$;}
\KwOut{$f_n^{k,l}\left(x\right)$.}
Prepare the initial state $\ket{\psi}\ket{-} = \frac{1}{\sqrt{n}}\sum_{i \in [n]}\ket{i}\ket{-}$.

Perform the unitary operation $\left(\inparen{S_0\otimes I}O_x\right)^{T-1}$ to $\ket{\psi}\ket{-}
$ and obtain the quantum state $\ket{\psi'}\ket{-}$, where $S_0 = 2\ket{\psi}\bra{\psi}-I$ and the oracle $O_x$ is defined in Equation (\ref{eq:oracle}).

Make a measurement to $\ket{\psi'}$ using project measurement $\inbrace{P_i}$, where $P_i = \ket{i}\bra{i}$ for $i\in [n]$. If the measurement result is $i$, then query $x_i$. Then we call Algorithm \ref{al:one_query} to distinguish the probability distribution of $x_i$ whether is $B\left(\sin^2 \left(2T-1\right)\theta_l\right)$ or $B\left(\sin^2 \left(2T-1\right)\theta_k\right)$. If the distribution is $B\left(\sin^2 \left(2T-1\right)\theta_l\right)$, we output $f_n^{k,l}\left(x\right) = 0$; otherwise, we output $f_n^{k,l}\left(x\right) = 1$. \label{step:measure_query}
\end{algorithm}

Next, we analyze the success probability of Algorithm \ref{al:T_queryKL1}. As shown in \cite{HSY+18}, $\ket{\psi}$ can be expressed as $\ket{\psi} = \cos \theta_x \ket{\alpha_{\bot}}+\sin \theta_x \ket{\alpha}$,
where $\theta_x = \arcsin{\sqrt{|x|/n}}$,
$\ket{\alpha_{\bot}} = \frac{1}{\sqrt{n-|x|}}\sum_{i:x_i = 0}\ket{i}$ and $\ket{\alpha} = \frac{1}{\sqrt{|x|}}\sum_{i:x_i = 1}\ket{i}$; after performing $T-1$ times operator $\inparen{S_0 \otimes I}O_x$, we have
\begin{equation*}
    \ket{\psi'} =  \cos \left(2T-1\right)\theta_x \ket{\alpha_{\bot}}+\sin \left(2T-1\right)\theta_x \ket{\alpha}.
\end{equation*}
Thus, if $|x| = l$, the distribution of $x_i$ in Step \ref{step:measure_query} of Algorithm \ref{al:T_queryKL1} is $B\left(\sin^2 \left(2T-1\right)\theta_l\right)$; if $|x| = k$, the distribution of $x_i$ is $B\left(\sin^2 \left(2T-1\right)\theta_k\right)$. By Lemma \ref{lemma:onequery}, the bias of the algorithm is 
\begin{equation*}
\begin{aligned}
\frac{1}{4}\inparen{\sin^2 \left(\left(2T-1\right)\theta_l\right) -\sin^2\left(\left(2T-1\right)\theta_k\right)} 
&\ge \frac{1}{4\pi}\left(2T-1\right)^2 \left(\sin^2 \theta_l - \sin^ 2 \theta_k\right) \\
&\ge \frac{1}{4\pi}\frac{l-k}{n}T^2\\
&= \Omega\inparen{\min \inbrace{\frac{l-k}{\sqrt{\left(n-k\right)l}}\cdot T, \frac{l-k}{n} \cdot T^2}},
\end{aligned}
\end{equation*}
where the second inequality follows by Claim $\ref{claim:sin_in}$.

ii) If $\frac{1}{4}\cdot \sqrt{\frac{n}{l}} < T < 3\pi \cdot \sqrt{\frac{n}{l}}$, we perform $T'$-query quantum algorithm as Algorithm \ref{al:T_queryKL1}, where $T' = \frac{1}{4}\cdot \sqrt{\frac{n}{l}}$. Let other queries be dummies. Then the bias of the algorithm is at least
\begin{equation*}
\frac{1}{4\pi} \frac{l-k}{n}T'^2 \ge \frac{1}{576\pi^3} \frac{l-k}{n}T^2 = \Omega\inparen{\min \inbrace{\frac{l-k}{\sqrt{\left(n-k\right)l}}\cdot T, \frac{l-k}{n} \cdot T^2}}.
\end{equation*}

iii) Suppose $3\pi \cdot \sqrt{\frac{n}{l}} \le T \le \frac{\sqrt{\left(n-k\right)l}}{4\sqrt{3}\pi\left(l-k\right)}$. Then $T \ge \frac{3\pi}{\sqrt{\frac{l}{n}}} \ge 3\pi/\theta_l$. 
Let $\Delta = \left(2T-1\right)\left(\theta_l-\theta_k\right)$, $\alpha = \left(2T-1\right)\theta_l$. 
By the mean value theorem, there exists some $a \in [\sqrt{\frac{k}{n}}, \sqrt{\frac{l}{n}}]$ such that 
\begin{equation*}
\begin{aligned}
\theta_l-\theta_k = \arcsin{\sqrt{\frac{l}{n}}}-\arcsin{\sqrt{\frac{k}{n}}} &= \frac{1}{\sqrt{1-a^2}}\left(\sqrt{\frac{l}{n}}-\sqrt{\frac{k}{n}}\right). \\
\end{aligned}
\end{equation*}
By $a \le \sqrt{\frac{l}{n}} = 1/2$ and \Cref{inequality_kl_upper}, we have
\begin{equation*}
\begin{aligned}
\arcsin{\sqrt{\frac{l}{n}}}-\arcsin{\sqrt{\frac{k}{n}}}
&\le 2\left(\sqrt{\frac{l}{n}}-\sqrt{\frac{k}{n}}\right) 
&\le \frac{2\left(l-k\right)}{\sqrt{\left(n-k\right)l}}.
\end{aligned}
\end{equation*}
Similarly, by $a \ge \sqrt{\frac{k}{n}} \ge 0$ and \Cref{inequality_kl_lower}, we have
\begin{equation*}
\begin{aligned}
\arcsin{\sqrt{\frac{l}{n}}}-\arcsin{\sqrt{\frac{k}{n}}}
&\ge \sqrt{\frac{l}{n}}-\sqrt{\frac{k}{n}} 
&\le \frac{\left(l-k\right)}{2\sqrt{2}\sqrt{\left(n-k\right)l}}.
\end{aligned}
\end{equation*}
Thus, we have 
\begin{equation}\label{Delta_upperbound}
\Delta \le 2T\left(\theta_l-\theta_k\right) = \frac{4\left(l-k\right)}{\sqrt{\left(n-k\right)l}}T \le \frac{1}{\sqrt{3}\pi}.
\end{equation}
and
\begin{equation}\label{Delta_lowerbound}
\Delta \ge  T\left(\theta_l-\theta_k\right)  \ge \frac{l-k}{4\sqrt{\left(n-k\right)l}}T.
\end{equation}
Moreover, the total variance distance between $B\left(\sin^2\left(2T-1\right)\theta_k\right)$ and $B\left(\sin^2\left(2T-1\right)\theta_l\right)$ is
\begin{equation*}
\begin{aligned}
d_{k,l} &=    \sin^2 \alpha - \sin^2 \left(\alpha-\Delta\right) \\ 
    &= \sin^2 \alpha-\left(\sin\alpha\cos \Delta - \cos \alpha \sin \Delta\right)^2 \\
    &= \sin^2 \alpha-\left(\sin^2\alpha\cos^2\Delta+\cos^2\alpha\sin^2\Delta-2\sin\alpha\cos\alpha\sin\Delta\cos\Delta\right) \\
    &= \sin^2 \alpha \sin^2\Delta-\cos^2\alpha \sin^2\Delta+\frac{1}{2}\sin 2\alpha\sin 2 \Delta \\
    &= \frac{1}{2}\sin 2\alpha\sin 2 \Delta-\cos 2\alpha \sin^2 \Delta.
\end{aligned}
\end{equation*}
Next, we discuss the following subcases: (1) If $\sin 2\alpha \ge 1/2$, by Fact \ref{sin_inequality}, \Cref{Delta_upperbound,Delta_lowerbound}, we have  
\begin{equation*}
\begin{aligned}
d_{k,l} &\ge \frac{1}{4}\sin 2\Delta-\frac{\sqrt{3}}{2}\sin^2 \Delta \\
&\ge \frac{1}{\pi}\Delta - \frac{\sqrt{3}}{2}\Delta^2 \\
&\ge \frac{1}{2\pi}\Delta \\
&\ge \frac{1}{8\pi}\frac{l-k}{\sqrt{\left(n-k\right)l}} \cdot T.\\ 
\end{aligned}
\end{equation*}
By Lemma \ref{lemma:onequery}, the bias of the algorithm is at least 
\begin{equation*}
\frac{1}{4}d_{k,l} = \Omega\inparen{\min \inbrace{\frac{l-k}{\sqrt{\left(n-k\right)l}}\cdot T, \frac{l-k}{n} \cdot T^2}}.
\end{equation*}
(2) If $\sin 2\alpha < 1/2$, 
let $T_- = T- \lceil \frac{\pi}{2\theta_l} \rceil$. Since $T \ge 3\pi/\theta_l$, we have $T_- \ge T/2$.  Note that for any $\theta$, there exist $\eta$ such that $\theta \le \eta \le \eta+2\pi/3\le \theta+2\pi$ and for any $\phi \in [\eta, \eta+2\pi/3]$, $\sin \phi \ge 1/2$. 
Let $\alpha_{t} = \left(2t-1\right)\theta_l$ for any integer $t$. Since $2\alpha_{T} - 2\alpha_{T_-} \ge 2\pi$ and  $2\alpha_{t+1}-2\alpha_t  = 4\theta_l < 2\pi/3$ for any $t \in [T_-, T)$, there exists integer $T' \in [T_-, T)$ such that $\sin 2\alpha_{T'}
\ge 1/2$.
We query $T'$ times and let other $T-T'$ queries be dummies. Accoring to the above subcase (1), the bias of the algorithm is at least
\begin{equation*}
\frac{1}{32\pi}\frac{l-k}{\sqrt{\left(n-k\right)l}}\cdot T' \ge \frac{1}{64\pi}\frac{l-k}{\sqrt{\left(n-k\right)l}}\cdot T= \Omega\inparen{\min \inbrace{\frac{l-k}{\sqrt{\left(n-k\right)l}}\cdot T, \frac{l-k}{n} \cdot T^2}}.
\end{equation*}

iv) If $T > \frac{1}{4\sqrt{3}\pi}\frac{\sqrt{\left(n-k\right)l}}{l-k}$, we only use $\frac{1}{4\sqrt{3}\pi}\frac{\sqrt{\left(n-k\right)l}}{l-k}$ queries and let other queries be dummies. According to the analysis of case iii), the bias of the algorithm is at least $\Omega\left(1\right)$.

Combining the above four cases, if $k \le n/2$ and $l-k = o\left(n\right)$, we have the following conclusion: If $T = O\inparen{\frac{\sqrt{\left(n-k\right)l}}{l-k}}$, then the optimal success probability $\gamma$ of $T$-query quantum algorithms to compute $\kl$ is at least $1/2+\Omega\inparen{\min \inbrace{\frac{l-k}{\sqrt{\left(n-k\right)l}}\cdot T, \frac{l-k}{n} \cdot T^2}}$. If $T = \Omega\inparen{\frac{\sqrt{\left(n-k\right)l}}{l-k}}$, then $\gamma = 1/2+\Omega\inparen{1}$.

If $k > n/2$ and $l-k = o(n)$, we can flip all bits of $x$ such that $n-l \le n-k \le n/2$ and run the above algorithm in the same way. Since $\sqrt{\left(n-\left(n-l\right)\right)\left(n-k\right)} = \sqrt{\left(n-k\right)l}$ and $(n-k)-(n-l) = l-k$, the corresponding bound is consistent with $k \le n/2$. Thus, the proof is finished.
\end{proof}
Finally, \Cref{th:kl_quantum} consists of \Cref{lemma:upperboundKL,lemma:lowerboundKL}, which is restated as follows:
\quanOptKl*

\subsection{The Classical Case}

\subsubsection{The Upper Bound}
In this section, we give the following lemma:
\begin{lemma}\label{lemma:kl_class_upper_bound}
For $T > 0$, the classical $T$-bias of $\kl$ satisfies that
\[\delta_C\inparen{\kl, T} =O\inparen{\min \inbrace{\IAdvf \cdot\sqrt{T}+\frac{T}{n}, \diff \cdot T}}.\]
\end{lemma}
\begin{proof}
By \Cref{proposition:classical_upperbound}, the success probability of $T$-query classical randomized algorithms to compute $\kl$ is at most $1/2+\alpha/2$ and 
\begin{equation}\label{eq:alpha_kl}
\begin{aligned}
\alpha &= \frac{1}{2}\sum_{i=0}^T\binom{T}{i}\cdot\left|\frac{\binom{n-T}{k-i}}{\binom{n}{k}}-\frac{\binom{n-T}{l-i}}{\binom{n}{l}}\right| \\
&= \frac{1}{2}\sum_{i=0}^T\left|\frac{\binom{k}{i}\binom{n-k}{T-i}}{\binom{n}{T}} - \frac{\binom{l}{i}\binom{n-l}{T-i}}{\binom{n}{T}}\right| \\
&= d_{TV}\inparen{H\left(n,k,T\right),H\left(n,l,T\right)}.
\end{aligned}
\end{equation}
By \Cref{fact:dH_dTV_1,lemma:Hellinger_kl}, we have
\begin{equation*}
d_{TV}\left(B\left(T,\frac{k}{n}\right),B\left(T,\frac{l}{n}\right)\right) \le T \cdot d_{TV}\left(B\left(\frac{k}{n}\right),B\left(\frac{l}{n}\right)\right) = \frac{l-k}{n}T,
\end{equation*}
and
\begin{equation*}
d_{TV}\left(B\left(T,\frac{k}{n}\right),B\left(T,\frac{l}{n}\right)\right) \le \sqrt{2}\sqrt{T}d_H\left(B\left(\frac{k}{n}\right),B\left(\frac{l}{n}\right)\right) = O\left(\IAdvf \sqrt{T}\right).
\end{equation*}
By \Cref{fact:dtv_BH}, we have 
\begin{equation*}
\begin{aligned}
d_{TV}\left(B\left(T,\frac{k}{n}\right), H\left(n,k,T\right)\right) &\le \frac{T-1}{n-1},\\
d_{TV}\left(B\left(T,\frac{l}{n}\right), H\left(n,l,T\right)\right) &\le \frac{T-1}{n-1}. \\
\end{aligned}
\end{equation*}
Thus,
\begin{equation}\label{eq:dTV_Hkl}
\begin{aligned}
d_{TV}\left(H\left(n,k,T\right),H\left(n,l,T\right)\right) 
&= O\inparen{\min\inbrace{\IAdvf \sqrt{T},\frac{l-k}{n}T}+\frac{T-1}{n-1}}  \\
&= O\inparen{\min\inbrace{\IAdvf \sqrt{T},\frac{l-k}{n}T}+\frac{T}{n}} \\
&= O\inparen{\min\inbrace{\IAdvf \sqrt{T}+\frac{T}{n},\frac{l-k}{n}T}}.
\end{aligned}
\end{equation}
By \Cref{eq:alpha_kl,eq:dTV_Hkl}, Lemma \ref{lemma:kl_class_upper_bound} is proven.
\end{proof}
\subsubsection{The Lower Bound}

We first show the following fact that will be used in the proof of Lemma \ref{lemma:kl_class_lower_bound}:
\begin{Fact}[Theorem 1 in \cite{zubkov2019recent}]\label{prop:zubkov}
Given the Bernoulli distributions $p = B(a)$ and $q = B(b)$, let $\epsilon = b-a$ and $v = \min\inbrace{a,1-a,b,1-b}$. If $v \in [\frac{2}{n},\frac{n-1}{2n})$, we have
\begin{equation*}
d_{TV}\left(p^{\otimes T},q^{\otimes T}\right) \ge \frac{e^{-4/nv\left(1-v\right)}}{2\left(1+2/T\right)}\sqrt{\frac{v}{1-v}}\left(\Phi\left(2\epsilon\sqrt{T}\right)-\frac{1}{2}\right),
\end{equation*}
where $\Phi\left(x\right) = \frac{1}{\sqrt{2\pi}}\int_{-\infty}^xe^{-u^2/2}du$ is a standard normal distribution function.
\end{Fact}
Then we prove the following lemma:
\begin{lemma}\label{lemma:kl_class_lower_bound}
If $T = O\inparen{\frac{(n-k)l}{(l-k)^2}}$, the classical $T$-bias of $\kl$ satisfies that
\[\delta_C\inparen{\kl, T} =O\inparen{\min \inbrace{\IAdvf \cdot\sqrt{T}+\frac{T}{n}, \diff \cdot T}}.\]
If $T = \Omega\inparen{\frac{(n-k)l}{(l-k)^2}}$, then $\delta_C\inparen{\kl, T}= \Omega(1)$.
\end{lemma}
\begin{proof}
To compute $f_n^{k,l}(x)$ $(k < l)$ is equivalent to determine whether $|x| = k$ or $l$. We give two different algorithms to distinguish $|x| = k$ from $|x| = l$ and analyze the bias of the algorithms. The first algorithm is as follows: We pad $2n$ 0's and $n$ 1's such that the problem is reduced to distinguish $|x'| = n+k$ or $n+l$ for $x' \in \inbrace{0,1}^{4n}$. Then we make $T$ queries to $x'$ uniformly and independently to determine whether the distribution is $p^{\otimes T}$ or $q^{\otimes T}$, where $p = B\left(\frac{n+k}{4n}\right)$ and $q = B\left(\frac{n+l}{4n}\right)$. Let $n' = 4n$. Regarding in Fact \ref{prop:zubkov}, we have $\epsilon = \frac{l-k}{4n}, v = \frac{n+k}{4n}$. Then we have 
\[
\begin{aligned}
\frac{1}{4} &\le v < \frac{2n-1}{4n} < \frac{n'-1}{2n'}, \\
\sqrt{\frac{v}{1-v}} &\ge \sqrt{1/3}, \\
n'v\left(1-v\right) &\ge \frac{3}{16}n' \ge \frac{3}{4}.
\end{aligned}
\]
Thus, $e^{-4/n'v\left(1-v\right)} \ge e^{-\frac{16}{3}}$. Since $2\left(1+2/T\right) \le 6$, we have
\begin{equation}\label{eq:Omega1}
    \frac{e^{-4/nv\left(1-v\right)}}{2\left(1+2/T\right)}\sqrt{\frac{v}{1-v}} \ge \frac{\sqrt{3}e^{-\frac{16}{3}}}{18} = \Omega\left(1\right).
\end{equation}
For $0 < x \le 1$, 
\begin{equation}\label{eq:Omega2}
\Phi\left(x\right) - \frac{1}{2} = \frac{1}{\sqrt{2\pi}}\int_{0}^xe^{-u^2/2}du \ge \frac{1}{\sqrt{2\pi}}\int_{0}^x e^{-1/2}du \ge \Omega\left(x\right). 
\end{equation}
Thus, if $T \le \frac{1}{4\epsilon^2} = \frac{4n^2}{\left(l-k\right)^2}$, substituting \Cref{eq:Omega1,eq:Omega2} into \Cref{prop:zubkov}, 
we have 
\begin{equation}\label{eq:kl_class_lower_1}
d_{TV}\left(p^{\otimes T},q^{\otimes T}\right) \ge \Omega\left(\epsilon \sqrt{T}\right) = \Omega\left(\frac{l-k}{n}\sqrt{T}\right).
\end{equation}
By \Cref{lemma:onequery} and Equation (\ref{eq:kl_class_lower_1}), the bias $\beta$ of the algorithm satisfies that
\begin{equation}\label{eq:lowerbound1}
    \beta = \Omega\inparen{d_{TV}\left(p^{\otimes T},q ^{\otimes T}\right)} =  \Omega\left(\frac{l-k}{n}\sqrt{T}\right).
\end{equation}
The second algorithm is as follows: we make $T$ queries uniformly and independently to determine whether the distribution is $p^{\otimes T}$ or $q^{\otimes T}$, where $p = B\left(\frac{k}{n}\right)$ and $q = B\left(\frac{l}{n}\right)$. By \Cref{fact:dH_dTV_1}, \Cref{claim_1epT} and \Cref{lemma:Hellinger_kl}, if $T \le \frac{1}{d_H^2\left(p^{\otimes T},q^{\otimes T}\right)} = O\inparen{ \frac{\left(n-k\right)l}{\left(l-k\right)^2}}$, then
\begin{equation*}
\begin{aligned}
    d_{TV}\left(p^{\otimes T},q ^{\otimes T}\right) 
    &\ge 1-\left(1-d_H^2\left(p,q\right)\right)^T \\
    &=\Omega\inparen{\frac{1}{e}d_H^2\left(p,q\right) \cdot T} \\
    &=\Omega\inparen{\IAdvfsqua T}.
\end{aligned}
\end{equation*}
Thus the bias $\beta$ of the second algorithm satisfies that
\begin{equation}\label{eq:lowerbound2}
    \beta = \Omega\inparen{d_{TV}\left(p^{\otimes T},q ^{\otimes T}\right)} =  \Omega\inparen{\IAdvfsqua T}.
\end{equation}
Combining \Cref{eq:lowerbound1,eq:lowerbound2}, the case for $T = O\inparen{ \frac{\left(n-k\right)l}{\left(l-k\right)^2}}$ is proven. If $T = \Omega\inparen{ \frac{\left(n-k\right)l}{\left(l-k\right)^2}}$, then there exists some constant $c > 0$ such that $T \ge c  \cdot\frac{\left(n-k\right)l}{\left(l-k\right)^2}$ and \[\delta_C\inparen{\kl, T} \ge  \delta_C\left(\kl,c \cdot \frac{(n-k)l}{(l-k)^2}\right) = \Omega(1).\]
\end{proof}
Finally, Theorem \ref{th:kl_class} is implied by Lemmas \ref{lemma:kl_class_upper_bound} and \ref{lemma:kl_class_lower_bound}, which is restated as follows:
\classOptKl*
Combining \Cref{th:kl_quantum,th:kl_class}, we recall and prove \Cref{cor:kl} as follows:
\corKl*
\begin{proof}
By \Cref{th:kl_quantum}, if a quantum algorithm using $T$ queries to compute $\kl$, then the success probability is at most $1/2+\beta$, where 
\[
\beta = \Theta\inparen{\min \inbrace{\frac{l-k}{\sqrt{(n-k)l}}\cdot T, \frac{l-k}{n} \cdot T^2}}.
\]
By \Cref{th:kl_class}, there exist classical algorithms using $T^2$ queries with success probability 
\[
    \frac{1}{2}+\Omega\inparen{\max\inbrace{\IAdvfsqua \cdot T^2, \diff \cdot T}} \ge
\frac{1}{2}+\Omega\inparen{\IAdvfsqua \cdot T^2} = \frac{1}{2}+ \Omega\inparen{\beta^2}.
\]
And there exist classical algorithms using $T^4$ queries with success probability
\[
    \frac{1}{2}+\Omega\inparen{\max\inbrace{\IAdvfsqua \cdot T^4, \diff \cdot T^2}} \ge
\frac{1}{2}+\Omega\inparen{\diff \cdot T^2} = \frac{1}{2}+ \Omega\inparen{\beta}.
\]
\end{proof}

\section{The Relation Between Quantum and Randomized Algorithms of Symmetric Boolean Functions for Arbitrarily Small Bias}
\label{sec:relationship_Q_R}
In this section, we prove Theorem \ref{th:relation_Q_R} and Corollary \ref{AA_conjecture}.
First we state \Cref{krav_transform}, which is needed to prove Theorem \ref{th:relation_Q_R}.
In \Cref{krav_transform}, since $g$ is a symmetric function, we let $\hat{g}(l) = \hat{g}(S)$ for any $|S| = l$ with a slight abuse of notation, where $\hat{g}(S)$ is the Fourier coefficients of $g$ for $S \subseteq [n]$. Moreover, $K_l(t,T)$ is the Kravchuk polynomial as \Cref{lemma:krav}.
\begin{restatable}{lemma}{kravtransform}
\label{krav_transform}
    Given a symmetric function $g:\inbrace{-1,1}^n \rightarrow [-1,1]$ such that $\deg(g) = d$, for any $d \le T \le n$ and $x \in \inbrace{-1,1}^n$, we have
    \[
g(x) = \mathbb{E}_{t\sim H(n,|x|,T)} \sum_{l=0}^d \hat{g}(l)\binom{n}{l}\frac{K_l(t,T)}{\binom{T}{l}},
    \]
    and
    \[
    \mathbb{E}_{t\sim B(T,\frac{1}{2})} \inparen{\sum_{l=0}^d \hat{g}(l)\binom{n}{l} \frac{K_l(t,T)}{\binom{T}{l}}}^2 \le 2.
    \]
\end{restatable}
\begin{proof}
Since $g:\inbrace{-1,1}^n \rightarrow [-1,1]$ is a symmetric function and $\deg(g) = d$,
 for any $d \le T \le n$ and $x \in \inbrace{-1,1}^n$, we have 
\begin{equation*}
\begin{aligned}
g(x) &= \sum_{S \subseteq [n]:|S| \le d}\hat{g}(S)x_S \\
&= \sum_{l=0}^d \hat{g}(l)\sum_{S \subseteq [n]:|S| = l}x_S \\
&= \sum_{l=0}^d \hat{g}(l)\frac{1}{\binom{n-l}{T-l}}\sum_{U \subseteq [n]:|U| = T}\sum_{S \subseteq U:|S| = l}x_S \\
&= \sum_{l=0}^d \hat{g}(l)\frac{\binom{n}{l}}{\binom{n}{T}\binom{T}{l}}\sum_{U \subseteq [n]:|U| = T}\sum_{S \subseteq U:|S| = l}x_S \\
&= \frac{1}{\binom{n}{T}}\sum_{U \subseteq [n]:|U| = T}\sum_{l=0}^d \hat{g}(l)\binom{n}{l}\frac{\sum_{S \subseteq U:|S| = l}x_S}{\binom{T}{l}}\\
&= \frac{1}{\binom{n}{T}}\sum_{t=0}^T\binom{|x|}{t}\binom{n-|x|}{T-t}\sum_{l=0}^d \hat{g}(l)\binom{n}{l}\frac{\sum_{i=0}^l\binom{t}{i}\binom{T-t}{l-i}(-1)^{i}}{\binom{T}{l}}\\
&= \mathbb{E}_{t\sim H(n,|x|,T)} \sum_{l=0}^d \hat{g}(l)\binom{n}{l}\frac{K_l(t,T)}{\binom{T}{l}}.\\
\end{aligned}
\end{equation*}
Let $c_{l} = \hat{g}\inparen{l}\binom{n}{l}$. By Fact \ref{fact:fourier_growth}, we have $|c_l| \le \frac{d^l}{l!}$. 
Then we have 
\begin{equation*}
\begin{aligned}
\mathbb{E}_{t\sim B(T,\frac{1}{2})} \inparen{\sum_{l=0}^d \hat{g}(l)\binom{n}{l} \frac{K_l(t,T)}{\binom{T}{l}}}^2 &=
\frac{1}{2^T}\sum_{t = 0}^T \binom{T}{t}\inparen{\sum_{l=0}^d c_l \frac{K_l(t,T)}{\binom{T}{l}}}^2  \\
&= \sum_{l=0}^d \frac{c_l^2}{\binom{T}{l}} \\
&\le \sum_{l=0}^d \frac{d^l}{l!}\cdot \frac{d^l}{l!} \cdot \frac{l!}{\prod_{i=0}^{l-1}T-i} \\
&\le \sum_{l=0}^d \frac{1}{l!}\cdot\frac{d^{2l}}{\prod_{i=0}^{l-1}T-i}\\
&\le \sum_{l=0}^d \prod_{i=0}^{l-1}\frac{d^2}{T-i}\\
&\le \sum_{l=0}^d \inparen{\frac{1}{2}}^l \\
&\le 2,
\end{aligned}
\end{equation*}
where the second equality comes from the orthogonality property of the Kravchuk polynomial (\Cref{lemma:krav}).
\end{proof}
Then we restate and prove Theorem \ref{th:relation_Q_R} as follows.
\relationQR*
\begin{proof}[Proof of \Cref{th:relation_Q_R}]
For any total symmetric Boolean function $f$ and $0 < \beta < 1/2,0 < \delta < 1$, let $\epsilon = 1/2-\beta$. Suppose there exists a quantum algorithm using $T$ queries to compute $f$ with success probability $1/2+\beta$. By \Cref{beals}, we have $\adeg_{\epsilon}(f) \le 2T$. Let $d = \adeg_{\epsilon}(f)$. Next, it suffices to prove there exists a randomized algorithm using $O(d^2)$ queries to compute $f$ with success probability $1/2+\Omega\inparen{\delta\beta^2}$  on a $1-\delta$ fraction of inputs. 

Since $d = \adeg_{\epsilon}(f)$, there exists a degree-$d$ symmetric function $f':\inbrace{0,1}^n \rightarrow [0,1]$ satisfying if $f(x) = 0$, then $f'(x) \le 1/2-\beta$; if $f(x) = 1$, then $f'(x) \ge 1/2+\beta$. It means that 
$\inparen{1-2f(x)}\inparen{1-2f'(x)} \ge 2\beta$. 
Let $h:\inbrace{0,1}^n \rightarrow [-1,1]$ be  defined as $h(x) = 1-2f'(x)$ and $g:\inbrace{-1,1}^n \rightarrow [-1,1]$ defined as $g(1-2x) = h(x)$ for any $x\in \B^n$.
Then $g$ is also a degree-$d$ symmetric function. 
By Lemma \ref{krav_transform}, for any $d \le T \le n$ and $x \in \inbrace{-1,1}^n$, we have 
\begin{equation*}
g(x) = \mathbb{E}_{t\sim H(n,|x|,T)} \sum_{l=0}^d \hat{g}(l)\binom{n}{l}\frac{K_l(t,T)}{\binom{T}{l}}.
\end{equation*}
Let 
\begin{equation}\label{eq:At}
 A_t = \sum_{l=0}^d \hat{g}\inparen{l}\binom{n}{l} \frac{K_l(t,T)}{\binom{T}{l}}.   
\end{equation}
Then for $x \in \B^n$, we have
$h(x) = \mathbb{E}_{t\sim H(n,|x|,T)}A_t$, 
where $|x|$ is the number of 1's in $x$.
For any $0 \le t \le T$,
let
\begin{equation}\label{eq:A_prime_t}
    A'_{t} = 
    \begin{cases}
    \min\inbrace{A_t, \frac{16}{\delta\beta}}, &\text{ if } A_t \ge 0,\\
    \max\inbrace{A_t, -\frac{16}{\delta\beta}}, &\text{ if } A_t < 0.\\
    \end{cases}
\end{equation}
Suppose $x$ follows the uniform distribution of $\inbrace{0,1}^n$. 
We give Algorithm \ref{random_algorithm} to compute $f(x)$ using $T = 2d^2+d$ queries.
\begin{algorithm}
\caption{A $T$-query quantum algorithm to compute $f(x)$.}
\label{random_algorithm}
Query $T$ distinct bits in $x$ uniformly and denote the number of 1's by $t$.

Compute the value of $A_t$ as \Cref{eq:At}.

Output $0$ with the probability $\frac{1}{2}(1+\frac{\delta\beta}{16}A'_{t})$ and output $1$ with the probability $\frac{1}{2}(1-\frac{\delta\beta}{16}A'_{t})$, where $A'_t$ is defined as \Cref{eq:A_prime_t}.
\end{algorithm}
The error analysis of Algorithm \ref{random_algorithm} is as follows. Given $x \in \{0,1\}^n$, let $h'(x) = \mathbb{E}_{t \sim H(n,|x|,T)}A'_t$. Then 
the probability that the algorithm outputs 0 is 
$\frac{1}{2}\inparen{1+\frac{\delta\beta}{16}h'(x)}$ and the probability that the algorithm outputs 1 is 
$\frac{1}{2}\inparen{1-\frac{\delta\beta}{16}h'(x)}$.
By Lemma \ref{krav_transform}, we have 
$\mathbb{E}_{t\sim B(T,\frac{1}{2})} A_t^2 \le 2$.
Since
\begin{equation}\label{eq:expt} 
\left|\mathbb{E}_{t\sim B(T,\frac{1}{2})} A_t\right| = \left|\frac{1}{2^n}\sum_{x \in \inbrace{0,1}^n}\mathbb{E}_{t\sim H\inparen{n,|x|,T}} A_t\right| = \left|\frac{1}{2^n}\sum_{x \in \inbrace{0,1}^n}h\inparen{x}\right| \le 1, \end{equation}
we have $\sigma^2\inparen{A_t} = \mathbb{E}A_t^2-  \inparen{\mathbb{E} A_t}^2  \le 2$ when $t$ follows the binomial distribution $B(T,\frac{1}{2})$. 
By Chebyshev's inequality,
we have
\begin{equation}\label{eq:cheb}
P\inparen{|A_{t}-\mathbb{E}A_t| \ge 2\delta} \le \frac{1}{\delta^2}.
\end{equation}
By \Cref{eq:A_prime_t}, we have
\begin{equation*}
    A'_{t} = 
    \begin{cases}
        A_{t},&\text{ if } |A_{t}| \le \frac{16}{\delta\beta}.\\
        -\frac{16}{\delta\beta},&\text{ if }A_{t} < -\frac{16}{\delta\beta}.\\
        \frac{16}{\delta\beta},& \text{ if }A_{t}> \frac{16}{\delta\beta}.
    \end{cases}
\end{equation*}
Thus if $|A_t| \le \frac{16}{\delta\beta}$, then $|A_t -A'_t| = 0$; if $|A_t| > \frac{16}{\delta\beta}$, then $|A_t-A'_t| = |A_t| - \frac{16}{\delta\beta}$.
Then we have
\begin{equation}\label{eq:Et1}
\mathbb{E}_t|A_t-A'_t|
= \mathbb{E}_{t:|A_t| \ge \frac{16}{\delta\beta}}\inparen{|A_t|-\frac{16}{\delta\beta}}.
\end{equation}
By \Cref{eq:expt}, we have $|\mathbb{E}A_t| \le 1$. Since $0 < \delta < 1, 0 < \beta < 1/2$, if $|A_t| > \frac{16}{\delta\beta}$, then $|A_t-\mathbb{E}A_t| \ge |A_t|-1 \ge \frac{15}{\delta\beta}$. Thus, we have
\begin{equation}\label{eq:Et2}
\begin{aligned}
\mathbb{E}_{t:|A_t| \ge \frac{16}{\delta\beta}}\inparen{|A_t|-\frac{16}{\delta\beta}} 
&\le \mathbb{E}_{t:|A_t-\mathbb{E}A_t| \ge \frac{15}{\delta\beta}}\inparen{|A_t|-\frac{16}{\delta\beta}} \\
&= \sum_{a = 0}^{\infty}\mathbb{E}_{t:\frac{15\cdot2^a}{\delta\beta} \le |A_t-\mathbb{E}A_t| \le \frac{15\cdot 2^{a+1}}{\delta\beta}}\inparen{|A_t|-\frac{16}{\delta\beta}} \\
&\le \sum_{a = 0}^{\infty}\mathbb{E}_{t:|A_t-\mathbb{E}A_t| \ge \frac{15\cdot 2^a}{\delta\beta}} \inparen{\frac{15\cdot 2^{a+1}}{\delta\beta}+1- \frac{16}{\delta\beta}} \\
&\le \sum_{a = 0}^{\infty}\mathbb{E}_{t:|A_t-\mathbb{E}A_t| \ge \frac{15 \cdot 2^a}{\delta\beta}} \frac{16\inparen{2^{a+1}-1}}{\delta\beta} \\
&\le \sum_{a=0}^{\infty}\inparen{\frac{\delta\beta}{15\cdot 2^{a-1}}}^2 \cdot \frac{16\inparen{2^{a+1}-1}}{\delta\beta}\\
    &= \frac{64}{225}\delta\beta \cdot \sum_{a = 0}^{\infty} \frac{2^{a+1}-1}{4^a}\\ 
    &\le \delta\beta,
\end{aligned}
\end{equation}
where the fourth inequality comes from \Cref{eq:cheb}.
Combining \Cref{eq:Et1,eq:Et2}, we have
\begin{equation}\label{eq:krav_expectation}
    \mathbb{E}_{t \sim B(T, \frac{1}{2})}|A_t-A_t'| 
    \le \delta\beta.
\end{equation}
For $x\in \B^n$, we have $h(x) = \mathbb{E}_{t \sim H(n,|x|,T)}A_t$ and $h'(x) = \mathbb{E}_{t \sim H(n,|x|,T)}A'_t$. 
Then we have
\[
\frac{1}{2^n}\sum_{x \in \inbrace{0,1}^n}|h(x)-h'(x)| = \frac{1}{2^n}\sum_{x \in \inbrace{0,1}^n}\mathbb{E}_{t \sim H(n,|x|,T)}|A_t-A_t'| = \mathbb{E}_{t \sim B(T, \frac{1}{2})}|A_t-A_t'| 
    \le \delta\beta.
\]
by Equation \ref{eq:krav_expectation}.
Thus, there are at least $1-\delta$ fractions of inputs $x$ such that  $|h(x)-h'(x)| \le \beta$. 
For such $x$, since $h(x)(1-2f(x)) \ge 2\beta$, we have $h'(x)(1-2f(x)) \ge \beta$. Therefore, if $f(x) = 0$, then $h'(x) \ge \beta$; if $f(x) = 1$, then $h'(x) \le -\beta$. Thus, for at least $1-\beta$ fractions of inputs, the bias of the algorithm is at least $\beta \cdot \delta\beta/16 = \delta\beta^2/16$.
\end{proof}

Before proving \Cref{AA_conjecture}, we give the following lemma:
\begin{lemma}[Amplification Lemma]\label{lemma:amplification}
For a Boolean function $f$, suppose there exists an algorithm $\mathcal{A}$ to compute $f$ such that for any $x$, $P\inparen{\mathcal{A}(x) = f(x)} = 1/2+\epsilon$. Then by running $T$ times the algorithm $\mathcal{A}$ independently, where $T \le \frac{1}{4\epsilon^2}$, we can amplify the success probability of the algorithm to $1/2+\Omega\inparen{\epsilon \sqrt{T}}$.
\end{lemma}
\begin{proof}
Let $p = B(1/2+\epsilon)$, $q = B(1/2-\epsilon)$. Each time we run Algorithm $\mathcal{A}$, if $f(x) = 0$, $\mathcal{A}(x)$ is a sample from distribution $q$; if $f(x) = 1$, $\mathcal{A}(x)$ is a sample from distribution $p$. 
Since $(1/2+\epsilon)-(1/2-\epsilon) = 2\epsilon$, we have $d_{TV}(p^{\otimes T},q^{\otimes T}) = \Omega\inparen{\epsilon \sqrt{T}}$ by Equation (\ref{eq:kl_class_lower_1}). Thus, after running $T$ time the algorithm $\mathcal{A}$ independently, we can distinguish $p$ from $q$ with bias $\Omega\inparen{\epsilon \sqrt{T}}$ by \Cref{lemma:onequery}. Equivalently, we can compute $f(x)$ with the success probability $1/2+\Omega\inparen{\epsilon \sqrt{T}}$.
\end{proof}
Then we restate and prove \Cref{AA_conjecture} as follows.
\AA*
\begin{proof}
Suppose there exists a quantum algorithm using $T$ queries to compute $f$ with bias $\beta$. By \Cref{th:relation_Q_R}, there exists a randomized algorithm $\mathcal{A}$ using $O(T^2)$ queries to compute $f$ with success probability at least $1/2+c_1\cdot\delta \beta^2$ on a $1-\delta$ fraction of inputs for some constant $c_2 > 0$. By \Cref{lemma:amplification}, after running $N$ times the algorithm $\mathcal{A}$ independently, we can amplify the success probability of the algorithm to $1/2+ c_1c_2\cdot\delta \beta^2 \sqrt{N}$ for some constant $c_2$. When 
\[
N = \frac{1}{c_1^2c_2^2} \cdot \inparen{\frac{\beta-\epsilon}{\delta\beta^2}}^2 = O\inparen{\frac{1}{\delta^2\beta^2}} = O\inparen{\frac{1}{\epsilon^2\delta^2}},
\]
we can obtain the desired success probability. Thus, the total number of queries is $O(T^2 \cdot N) = O\inparen{T^2/(\epsilon^2\delta^2)}$.
\end{proof}

\section{The Relation Between Complexity Measures of Partial Symmetric Boolean Functions}
\label{sec:relationship}
This section investigates the relation between several fundamental complexity measures of partial symmetric Boolean functions. Section \ref{subsec:qf_adeg} gives the relation between the quantum
query complexity and the approximate degree for errors close to $1/2$. Section \ref{subsec:qef_degree} shows exact quantum query complexity is at most quadratic of the degree. Section \ref{subsec:tight_complexity_measure} exhibits tight characterizations of several complexity measures and shows the relation between these complexity measures.  Section \ref{subsec:exp_gap} demonstrates an exponential gap between exact quantum query complexity and randomized query complexity.

\subsection{The Relation Between Quantum Query Complexity and Approximate Degree for Arbitrarily Small Bias}\label{subsec:qf_adeg}
In this section, we restate and 
give the proof of \Cref{th:poly_Q_adeg}. 
\adegQ*
\begin{proof}
Given a (possibly partial) $n$-bit symmetric Boolean function $f:D \rightarrow \B$, where $D \subseteq \B^n$ and $f(x) = f(n-x)$ for any $x \in D$. For $0 < \epsilon < 1/2$, let $T = \deg_{\epsilon}(f)$ and $\beta = 1/2-\epsilon$. Same as the proof of \Cref{th:relation_Q_R}, for any function $f'$ that approximates $f$ with error $\epsilon$, we have
$(1-2f(x))(1-2f'(x)) \ge 2\beta$.

Let $g: [-1,1] \rightarrow [-1,1]$ be defined as $g\left(1-2|x|/n\right) = 1-2f(x)$ for any $x \in D$. Since $f(x) = f(n-x)$, $g$ is an even function. Assume function $h:[-1,1] \rightarrow [-1,1]$ is the optimal approximation polynomial of $g$ with degree $T$. Then $g(x)h(x) \ge 2\beta$ and $h$ is also an even function. Thus, $h(x)$ can be expressed as  $\sum_{i=0}^{\lceil T/2 \rceil} a_i T_{2i}(x)$ for any $x \in [-1,1]$, where $T_{2i}(x)$ is the Chebyshev polynomial of degree $2i$ and $T_{2i}(\cos \eta) = \cos 2i\eta$ for any $\eta \in [0,\pi]$. Furthermore, we have $h(\cos\eta) = \sum_{i=0}^{\lceil T/2 \rceil} a_{i} \cos 2i\eta$ for any $\eta \in [0,\pi]$. Let $\cos \eta_x = 1-2|x|/n$. Then $\left(1-2f(x)\right) h\left(\cos \eta_x\right) \ge 2\beta$ and 
\begin{equation}\label{eq:hDM}
\begin{aligned}
h(\cos \eta_x) &= \sum_{i=0}^{\lceil T/2 \rceil}a_{i} \cos 2i\eta_x \\
&= \sum_{i:a_{i} \ge 0}a_{i} \left(2\cos^2 i\eta_x-1\right) + \sum_{ i:a_{i} < 0}a_{i} \left(1-2\sin^2 i\eta_x\right) \\
&= \left(\sum_{i:a_{i} \ge 0}2a_{i} \cos^2 i\eta_x - \sum_{ i:a_{i} < 0}2a_{i} \sin^2 i\eta_x\right)+\left(\sum_{i:a_{i} < 0}a_{i}-\sum_{i:a_{i} \ge 0}a_{i}\right) \\
&= \Delta_x-M,
\end{aligned}
\end{equation}
where 
$
\Delta_x = 2\inparen{\sum_{ a_{i} \ge 0}a_{i} \cos^2 i\eta_x - \sum_{ a_{i} < 0}a_{i} \sin^2 i\eta_x}$
and $M =  \sum_{i=0}^{\lceil T/2 \rceil} |a_i|$. By \Cref{lemma:parseval}, 
$\sum_{i=0}^{\lceil T/2 \rceil} a_i^2 \le \frac{2}{\pi}
\int_{-1}^1 \frac{1}{\sqrt{1-x^2}} dx = 2.$
Thus, 
\begin{equation}\label{eq:M}
M \le \sqrt{2\lceil T/2 \rceil+1} \le \sqrt{2(T+1)}.
\end{equation}
Let $\ket{\psi} = \frac{1}{\sqrt{n}}\sum_{i\in [n]} \ket{i}\ket{-}$.
Then 
$\langle \psi|O_x | \psi\rangle  =  1- 2|x|/n = \cos \eta_x$.   
As a result, there exists a state $\ket{\psi^{\bot}}$ such that $\langle \psi | \psi^{\bot} \rangle = 0$ and 
$O_x \ket{\psi} = \cos \eta_x \ket{\psi}+\sin \eta_x \ket{\psi^{\bot}}$.
For the following reflection operation
\begin{equation}\label{ref_S0S1}
\begin{aligned}
S_{0} &= 2\ket{\psi}\bra{\psi}-I,   
S_{1} = 2O_x\ket{\psi}\bra{\psi}O_x-I = O_xS_0O_x, \\
\end{aligned}
\end{equation}
we have 
\begin{equation}\label{eq:S_1S_0_r}
\begin{aligned}
S_1S_0\ket{\psi} &= \cos 2\eta_x\ket{\psi}+\sin 2\eta_x\ket{\psi^{\bot}},\\
S_1S_0\ket{\psi^{\bot}} &= -\sin 2\eta_x\ket{\psi}+\cos 2\eta_x\ket{\psi^{\bot}}.
\end{aligned}
\end{equation}
Let $R_0$ be the corresponding controlled  operation of $S_0$, i.e., for any $\ket{\phi}$,
\begin{equation}\label{eq:reflection_quantum_adeg}
    R_0\ket{\phi}\ket{+} = \ket{\phi}\ket{+}, 
    R_0\ket{\phi}\ket{-} = (S_0\ket{\phi})\ket{-}. 
\end{equation}
Let $\ket{\pm_i} = \underbrace{\ket{-}\cdots \ket{-}}_{i}\underbrace{\ket{+}\cdots \ket{+}}_{\lceil T/2 \rceil-i}$.
If $a_i \ge 0$, let 
\begin{equation*}
\begin{aligned}
P_i^+ &= (\ket{\psi}\bra{\psi}) \otimes (\ket{\pm_i}\bra{\pm_i}),
P_i^- = (I-\ket{\psi}\bra{\psi})\otimes (\ket{\pm_i}\bra{\pm_i}). 
\end{aligned}
\end{equation*}
If $a_i < 0$, let 
\begin{equation*}
\begin{aligned}
P_i^- &= (\ket{\psi}\bra{\psi}) \otimes (\ket{\pm_i}\bra{\pm_i}),
P_i^+ = (I-\ket{\psi}\bra{\psi})\otimes (\ket{\pm_i}\bra{\pm_i}). 
\end{aligned}
\end{equation*}
Let 
 $   P_0 = \sum_i P_i^+,
    P_1 = \sum_i P_i^-. $
Then $P_0 + P_1 = I$. 
Let
$\alpha_i = \sqrt{\frac{|a_i|}{M}}.$
Then $\sum_i \alpha_i^2 = 1$. We give Algorithm \ref{algorithm_small_bias} to compute $f(x)$ and analyze the success probability of the algorithm as follows.

\begin{algorithm}
\label{algorithm_small_bias}
\caption{A $T$-query quantum algorithm to compute $f(x)$.}
Prepare the initial state $\sum_{i=0}^{\lceil T/2 \rceil} \alpha_i \ket{\psi}\ket{\pm_i}$, which consists of the first qudit and $\lceil T/2 \rceil$ ancillary qubits, where $\alpha_i,\ket{\psi},\ket{\pm_i}$ are defined on Page 11. 

For $i = 1$ to $\lceil T/2 \rceil$, we perform unitary operation $(O_x \otimes I)R_0$ in the first qudit and the $i$-th ancillary qubit, where $R_0$ is given in \Cref{eq:reflection_quantum_adeg}.

Perform the project measurement $\inbrace{P_0,P_1}$ defined on Page 11 to the final state and output the measurement result.
\end{algorithm}

Since $R_0$ is the corresponding controlled reflection operation of $S_0$, the final state after performing Step 2 of Algorithm \ref{algorithm_small_bias} is \[
\sum_{i=0}^{\lceil T/2 \rceil} \alpha_i\left(\underbrace{(O_xS_0)\cdots (O_xS_0)}_{i \text{ times}}\ket{\psi}\right)\ket{\pm_i}.\]
If $i$ is even, then 
\begin{equation*}
\begin{aligned}
\left(\underbrace{(O_xS_0)\cdots (O_xS_0)}_{i \text{ times}}\ket{\psi}\right)\ket{\pm_i} 
&= \left(\underbrace{(O_xS_0O_xS_0)\cdots (O_xS_0O_xS_0)}_{i/2 \text{ times}}\ket{\psi}\right)\ket{\pm_i} \\
&= \left(\underbrace{(S_1S_0)\cdots (S_1S_0)}_{i/2 \text{ times}}\ket{\psi}\right)\ket{\pm_i} \\
&= \inparen{\cos i\eta_x \ket{\psi}+\sin i\eta_x \ket{\psi^{\bot}}}\ket{\pm_i}, 
\end{aligned}
\end{equation*}
where the second equality comes from $S_1 = O_xS_0O_x$ and the third equality comes from \Cref{eq:S_1S_0_r}.
Similarly, if $i$ is odd, by \Cref{eq:S_1S_0_r} and $S_1 = O_xS_0O_x$, we have 
\begin{equation*}
\begin{aligned}
\left(\underbrace{(O_xS_0)\cdots (O_xS_0)}_{i \text{ times}}\ket{\psi}\right)\ket{\pm_i}
&= \inparen{\underbrace{(O_xS_0)\cdots (O_xS_0)}_{i-1 \text{ times}}\inparen{\cos \eta_x \ket{\psi}+\sin \eta_x\ket{\psi^{\bot}}}}\ket{\pm_i} \\
&=\inparen{\underbrace{(S_1S_0)\cdots (S_1S_0)}_{(i-1)/2} \inparen{\cos \eta_x \ket{\psi}+\sin \eta_x\ket{\psi^{\bot}}}}\ket{\pm_i} \\
&= \inparen{\cos i\eta_x \ket{\psi}+\sin i\eta_x\ket{\psi^{\bot}}}\ket{\pm_i}. 
\end{aligned}
\end{equation*}
Thus, after performing Step 2 of Algorithm \ref{algorithm_small_bias}, the final state is 
\begin{equation*}
\sum_{i=0}^{\lceil T/2 \rceil} \alpha_i \left(\cos i\eta_x \ket{\psi}+ \sin i\eta_x\ket{\psi^{\bot}}\right) \ket{\pm_i}.
\end{equation*}
By \Cref{eq:hDM}, the probability that the measurement result is $0$ is 
\begin{equation*}
\begin{aligned}
p_x &= \sum_{i:a_i \ge 0} \alpha_i^2 \cos^2 i\eta_x + \sum_{i:a_i < 0} \alpha_i^2 \sin^2 i\eta_x \\
&= \frac{1}{M}\inparen{\sum_{i:a_i \ge 0} a_i \cos^2 i\eta_x - \sum_{i:a_i < 0} a_i \sin^2 i\eta_x} \\
&=\frac{\Delta_x}{2M} \\
&= \frac{1}{2}+ \frac{h(\cos \eta_x)}{2M},
\end{aligned}
\end{equation*}
and the probability that the algorithm outputs 1 is $1/2 - h(\cos \eta_x)/(2M)$. Since $(1-2f(x))h(\cos \eta_x) 
 \ge 2\beta$, the probability that the algorithm outputs $f(x)$ is at least $1/2+\beta/M$. By \Cref{eq:M}, 
 the bias of the algorithm is at least 
  $\beta/M \ge \beta/\sqrt{2T+2}$.
Then we can amplify the success probability to $1/2+\beta$ by running $O(T)$ times Algorithm \ref{algorithm_small_bias}  repetitively (See \Cref{lemma:amplification}). Thus, there exists a quantum algorithm using $O(T^2)$ queries to with success probability $1-\epsilon$, which implies 
 $   Q_{\epsilon}(f) = O\inparen{\deg_{\epsilon}(f)^2}$.
\end{proof}
\begin{remark}
We conjecture that $f(x)=f(n-x)$ is not a necessary condition. If an $n$-bit (possibly partial) symmetric Boolean function $f$ satisfies that $f(x) \neq f(n-x)$ for some $x \in D$, we can define a new $2n$-bit Boolean function $f^*$ such that 
\[
f^*(x) =\begin{cases}
f(x),&\text{ if }|x| \le n,\\
f(2n-x),&\text{ if }|x| > n. \\
\end{cases}
\]
Then $f^*$ satisfies that $f^*(x) = f^*(2n-x)$. Although we can run Algorithm \ref{algorithm_small_bias} to $f^*$, we do not know how to relate $\adeg_{\epsilon}(f^*)$ and $\adeg_{\epsilon}(f)$ for any $\epsilon$ arbitrarily close to 1/2. Thus, the query complexity of the algorithm is not promised. We leave this case as an open problem.  
\end{remark}

\subsection{The Relation Between Exact Quantum Query Complexity and Degree}\label{subsec:qef_degree}

In this section, we prove Theorem \ref{th:polyrelation1}, i.e., for any partial symmetric Boolean function $f$, we have $Q_E(f) = O(\deg(f)^2)$.
\begin{proof}[Proof of Theorem \ref{th:polyrelation1}]
Without loss of generality, any $n$-bit partial symmetric Boolean function $f:D \rightarrow \B$ can be defined as
\begin{equation*}
f(x) = 
\begin{cases}
0 & \text{ if } |x| \in S_0, \\
1 & \text{ if } |x| \in S_1,
\end{cases}
\end{equation*}
where $S_0, S_1 \subseteq [n]$.
By Fact \ref{symmetrization}, there
exists a single-variate polynomial $F:[0,n] \rightarrow \mathbb{R}$ such that $F(|x|) = f(x)$ and $\deg(F) \le \deg(f)$. Since $F$ has at least $|S_0|$ zeros and $1-F$ has at least $|S_1|$ zeros, we have $\deg(f) \ge \deg(F) \ge \max\inbrace{|S_0|, |S_1|}$. By  \Cref{beals}, we have 
\begin{equation}\label{eq:general_deg_lower_bound_1}
    Q_E(f) \ge \frac{1}{2}\deg(f ) \ge \frac{1}{2}\max\inbrace{|S_0|, |S_1|}.
\end{equation}
By Lemma \ref{lemma:deg_lowerbound}, if $f(k) \neq f(l)$ and $k < l$, then 
$\deg_{\epsilon}(f) = \Omega\inparen{\frac{\sqrt{\beta(n-k)l}}{l-k}}$, where $\beta = \frac{1}{2}-\epsilon$. Thus,
\begin{equation}\label{eq:general_deg_lower_bound_2}
\deg(f) = \deg_{0}(f) = \Omega\inparen{\frac{\sqrt{(n-k)l}}{l-k}}.
\end{equation}
By Fact \ref{fact:kl_exact_upper_bound}, given an $n$-bit Boolean string $x$ and promising that $|x| = k$ or $l$, there exists a quantum algorithm $\mathcal{A}_{k,l}$ to distinguish $|x| = k$ from $|x| = l$ with $O\left(\frac{\sqrt{(n-k)l}}{l-k}\right)$ queries exactly. 
Similar to the idea of the exact quantum algorithm to compute $\fnk$
in Section
\ref{subsubsec:fnk_quantumlowerbound}, we give the following algorithm to compute $f(x)$ for $x \in D$ exactly. Each time we choose some $k \in S_0,l \in S_1$ and run $\mathcal{A}_{k,l}$. If the output of $\mathcal{A}_{k,l}$ is $|x| = l$, then it means $|x| \neq k$; if the output of $\mathcal{A}_{k,l}$ is $|x| = k$, then it means $|x| \neq l$. In whichever case, we can exclude one possible value of $|x|$. Then we repeat the above procedure until there is only one possible value of $|x|$.
After at most $|S_0|+|S_1|-1$ iterations, we determine the value of $|x|$ with certainty. By \Cref{eq:general_deg_lower_bound_1,eq:general_deg_lower_bound_2}, the total number of queries is at most 
\begin{equation*}
\begin{aligned}
\left(|S_0|+|S_1|-1\right)\cdot O\inparen{\frac{\sqrt{(n-k)l}}{l-k}} &\le 2\max\inbrace{|S_0|, |S_1|}\cdot O\inparen{\frac{\sqrt{(n-k)l}}{l-k}} \\
&= O\inparen{\deg(f)^2}.
\end{aligned}
\end{equation*}
\end{proof}

\subsection{The Relation Between Quantum Query Complexity, Block Sensitivity and Fractional Block Sensitivity}\label{subsec:tight_complexity_measure}
In this section, we first consider fractional block sensitivity and block sensitivity of partial symmetric Boolean functions. We give the following lemma:
\begin{lemma}\label{lemma:bs_fbs}
For any partial symmetric Boolean function $f$, we have 
\begin{equation*}
\bs(f) = \Theta\left(\fbs(f)\right) = \Theta\inparen{\max\limits_{k<l:f(k)\neq f(l)} \frac{n}{l-k}}.
\end{equation*}
\end{lemma}
\begin{proof}
Recall that
\begin{align*}
\FC(f,x) = \min \sum_{i \in [n]} w_i \hspace{1cm} 
\text{subject to: } &\forall B \text{ s.t. } x,x^B \in D, f(x) \neq f(x^B): \sum_{i \in B} w_i \geq 1, \\
&\forall i \in [n]: 0 \leq w_i \leq 1.
\end{align*}
And $\fbs(f)=\FC(f)=\max_x\FC(f,x)$ by Fact~\ref{bs_fbs_FC}. For any $i \in [n]$, let $w_i = \max\limits_{k<l:f(k)\neq f(l)}\frac{1}{l-k}$. Notice that for any block $B$ sensitive in $x$, it satisfies that $|B|\geq\min_{k<l:f(k)\neq f(l)}(l-k)$. Then $\{w_i\}$ satisfies the constraint in the definition of $\FC(f,x)$ for any $x \in D$. Thus, for any $x\in D$, $\FC(f,x) \le \max\limits_{k<l:f(k)\neq f(l)}\frac{n}{l-k}$, which implies  
\begin{equation}\label{eq:FC_upperbound}
\FC(f) = \max_x \FC(f,x) \le \max\limits_{k<l:f(k)\neq f(l)}\frac{n}{l-k}.
\end{equation}
Combining $\bs(f) \le \fbs(f) = \FC(f)$ (\Cref{bs_fbs_FC}) and Lemma \ref{lemma:bs}, we have 
\[
\max\limits_{k<l:f(k)\neq f(l)}\left\lfloor\frac{n}{2(l-k)}\right\rfloor \le \bs(f) \le \fbs(f) = \max\limits_{k<l:f(k)\neq f(l)}\frac{n}{l-k}.
\]
Thus $\bs(f) = \Theta(\fbs(f)) = \Theta\inparen{\max\limits_{k<l:f(k)\neq f(l)} \frac{n}{l-k}}$.
\end{proof}

Next, we give a tight bound of quantum query complexity and approximate degree for any partial symmetric Boolean function by the following lemma:
\begin{lemma}\label{lemma:Q_tight_bound}
For any partial symmetric Boolean function $f$, we have
\[
Q(f) = \Theta\inparen{\adeg(f)} = \Theta\inparen{\max\limits_{k < l:f(k)\neq f(l)}\frac{\sqrt{(n-k)l}}{l-k}}.
\]
\end{lemma}
\begin{proof}
We first prove the upper bound. Given a partial symmetric Boolean function $f: D \rightarrow \B$, where $D \subseteq \B^n$ and an input $x \in D$, we give the following algorithm to compute $f(x)$: First, 
we use quantum approximate counting algorithm in Fact \ref{amplitude_estimation} with $6\pi \cdot\max\limits_{k,l:f(k)\neq f(l)}\frac{\sqrt{(n-k)l}}{|l-k|}$ queries to obtain the estimation value of $|x|$, denoted by $\Tilde{t}$. Let $W = \inbrace{|y|:y\in D}$. Then we select $w \in W$ such that $|w-\Tilde{t}|$ is minimal and output $f(w)$. 

The error analysis of the algorithm is as follows. Let $S_x = \inbrace{|k|:k \in W, f(k) = f(x)}$ and $\overline{S_{x}} = W\setminus S_x$.
%\inbrace{|k|:k \in S, f(k) \neq f(x)}$. 
Suppose $|x| = t$. By Fact \ref{amplitude_estimation}, with probability at least $\frac{8}{\pi^2}$, we have
\begin{equation*}
\begin{aligned}
    |\Tilde{t} - t| &\le \frac{\sqrt{\frac{t}{n}\inparen{1-\frac{t}{n}}}}{3\max\limits_{k, l:f(k)\neq f(l)}\frac{\sqrt{(n-k)l}}{|l-k|}}+\frac{1}{36\max\limits_{k,l:f(k)\neq f(l)}\frac{(n-k)l}{(l-k)^2}}. \\
\end{aligned}
\end{equation*}
Since 
\begin{equation*}
\begin{aligned}
    \frac{\sqrt{\frac{t}{n}\inparen{1-\frac{t}{n}}}}{\max\limits_{k,l:f(k)\neq f(l)}\frac{\sqrt{(n-k)l}}{|l-k|}} &= \min\limits_{k,l:f(k)\neq f(l)} \frac{|l-k|}{n}\frac{\sqrt{(n-t)t}}{\sqrt{(n-k)l}} \\
    &\le \min\inbrace{\min\limits_{k<t:f(k)\neq f(t)} \frac{t-k}{n}\frac{\sqrt{(n-t)t}}{\sqrt{(n-k)t}}, \min\limits_{t<l:f(t)\neq f(l)} \frac{l-t}{n}\frac{\sqrt{(n-t)t}}{\sqrt{(n-t)l}}} \\
    &\le \min\inbrace{\min\limits_{k<t:f(k)\neq f(t)} \frac{t-k}{n}, \min\limits_{t<l:f(t)\neq f(l)} \frac{l-t}{n}} \\
    &= \min\limits_{k:f(k)\neq f(t)} \frac{|k-t|}{n},
\end{aligned}
\end{equation*}
and
\begin{equation*}
\begin{aligned}
    \frac{1}{\max\limits_{k,l:f(k)\neq f(l)}\frac{(n-k)l}{(l-k)^2}} &=  \min\limits_{k,l:f(k)\neq f(l)}\frac{(l-k)^2}{(n-k)l} \\
    &= \min\limits_{k<l:f(k)\neq f(l)}\frac{(l-k)^2}{(n-k)l}\\
    &= \min\limits_{k<l:f(k)\neq f(l)}\frac{l-k}{\max\inbrace{n-k,l}}\frac{l-k}{\min\inbrace{n-k,l}} \\
    &\le \min\limits_{k<l:f(k)\neq f(l)}\frac{l-k}{\frac{1}{2}n} \\
    &\le \min\inbrace{ \min\limits_{k<t:f(k)\neq f(t)}\frac{2(t-k)}{n},   \min\limits_{t<l:f(t)\neq f(l)}\frac{2(l-t)}{n}    } \\
    &= \min\limits_{k:f(k)\neq f(t)}\frac{2|t-k|}{n},
\end{aligned}
\end{equation*}
we have
\begin{equation*}
    |\Tilde{t}-t|\le \min\limits_{k:f(k)\neq f(t)}\frac{1}{3}\frac{|t-k|}{n}+\frac{1}{18}\frac{|t-k|}{n} = \min\limits_{k:f(k)\neq f(t)}\frac{7}{18}\frac{|t-k|}{n}.
\end{equation*}
As a result, with the probability at least $\frac{8}{\pi^2}$, for any $k \in \overline{S_{x}}$ we have $|t-\Tilde{t}| \le \frac{7}{18}|t-k|$, which implies $|w-\Tilde{t}| < |k-\Tilde{t}|$. Thus $w \notin \overline{S_{x}}$ and $f(w) = f(x)$. Therefore, the success probability of the algorithm is at least $\frac{8}{\pi^2}$, which implies 
\begin{equation}\label{eq:Q_upper_bound}
    Q(f) = O\inparen{\max\limits_{k,l:f(k)\neq f(l)}\frac{\sqrt{(n-k)l}}{|l-k|}} = O\inparen{\max\limits_{k < l:f(k)\neq f(l)}\frac{\sqrt{(n-k)l}}{l-k}}.
\end{equation}
Moreover, for the lower bound, for any $k < l $ and $f(k) \neq f(l)$, we have 
\begin{equation}\label{eq:Q_lower_bound}
Q(f) \ge \frac{1}{2}\adeg(f) \ge \frac{1}{2}\adeg(\kl) = \Omega\inparen{\frac{\sqrt{(n-k)l}}{l-k}}
\end{equation}
by Lemma \ref{lemma:deg_lowerbound}. Note that this lower bound on the quantum query complexity was also obtained by \cite{Choi12, AA14} using the quantum adversary method. Combining \Cref{eq:Q_upper_bound,eq:Q_lower_bound}, we have 
$Q(f) = \Theta\inparen{\adeg(f)} = \Theta\inparen{\max\limits_{k < l:f(k)\neq f(l)}\frac{\sqrt{(n-k)l}}{l-k}}$.
\end{proof}
Finally, \Cref{th:polyrelation2} consists of \Cref{lemma:bs_fbs,lemma:Q_tight_bound}. Since $\sqrt{(n-k)l} \le n$, \Cref{th:polyrelation2} implies \Cref{th:poly_Q_bs}. We recall  \Cref{th:polyrelation2,th:poly_Q_bs} as follows:
\bsfbs*
\Qbs*

\subsection{Exponential Gap Between Exact Quantum Query Complexity and Randomized Query Complexity}\label{subsec:exp_gap}
In this section, we prove Theorem \ref{th:exprelation}, i.e., there exists a partial symmetric Boolean function $f$ such that $Q_E(f) = \Omega(n)$ and $R(f) = O(1)$.
\begin{proof}[Proof of Theorem \ref{th:exprelation}]
Let 
\begin{equation*} f(x) = 
\begin{cases}
0 & \text{ if } |x| \le n/2, \\
1 & \text{ if } |x| = n.
\end{cases}
\end{equation*}
Then $\deg(f) \ge n/2+1$ and $Q_E(f) \ge \deg(f)/2 = n/4+1/2$.
On the other hand, by \Cref{eq:Q_upper_bound}, we have $Q(f) = O\inparen{\max\limits_{k<l:f(k)\neq f(l)}\frac{\sqrt{(n-k)l}}{l-k}} =
O\inparen{\frac{\sqrt{(n-n/2)n}}{n-n/2}} = O(1)$.
Since $R(f) = O\left(Q(f)^2\right)$ \cite{AA14}, we have $R(f) =O(1)$.
\end{proof}

\section{Conclusion}\label{sec:conclusion}
This paper analyzes the quantum advantage of computing two fundamental partial symmetric Boolean functions by studying the optimal success probability of $T$-query quantum and randomized algorithms. Moreover, we analyze the relation between the number of queries and the bias of quantum and randomized algorithms to compute total symmetric Boolean functions when the bias of the algorithms can be arbitrarily small.
Furthermore, we show the relation of several fundamental complexity measures of partial symmetric Boolean functions. 
We leave the fine-grained Watrous conjecture as an open problem for further study.

\newpage

\bibliography{references}
\end{document}